\documentclass[reqno,12pt,letterpaper]{amsart}
\usepackage{amsmath,amssymb,amsthm,graphicx,mathrsfs,url,bbm,mathtools,cases,enumitem}
\usepackage[usenames,dvipsnames]{color}
\usepackage[colorlinks=true,linkcolor=Red,citecolor=Green]{hyperref}
\usepackage{amsxtra}
\usepackage{wasysym} 
\usepackage{graphicx}
\usepackage{subcaption}
\def\arXiv#1{\href{http://arxiv.org/abs/#1}{arXiv:#1}}

\usepackage{mathtools}

\def\?[#1]{\textbf{[#1]}\marginpar{\Large{\textbf{??}}}}

\def\smallsection#1{\smallskip\noindent\textbf{#1}.}
\let\epsilon=\varepsilon 

\newcommand{\RR}{{\mathbb R}}

\newcommand{\CC}{{\mathbb C}}

\newcommand{\ZZ}{{\mathbb Z}}

\newcommand{\ach}{{\operatorname{ac}}}

\newcommand{\bfk}{{\mathbf k}}
\newcommand{\bfa}{{\mathbf a}}

\newcommand{\bfn}{{\mathbf n}}

\newtheorem{theo}{Theorem}
\newtheorem{prop}{Proposition}[section]	
\newtheorem{defi}[prop]{Definition}

\newtheorem{assumption}{Assumption}
\newtheorem{lemm}[prop]{Lemma}
\newtheorem{corr}[prop]{Corollary}
\newtheorem{rem}{Remark}

\numberwithin{equation}{section}

\DeclareMathOperator{\diag}{diag}

\DeclareMathOperator{\Spec}{Spec}

\let\Im=\Imag

\let\Re=\Real
\DeclareMathOperator{\sgn}{sgn}

\DeclareMathOperator{\WF}{WF}

\DeclareMathOperator{\tr}{tr}

\def\indic{\operatorname{1\hskip-2.75pt\relax l}}
\usepackage{scalerel}

\newcommand\reallywidehat[1]{\arraycolsep=0pt\relax%
\begin{array}{c}
\stretchto{
  \scaleto{
    \scalerel*[\widthof{\ensuremath{#1}}]{\kern-.5pt\bigwedge\kern-.5pt}
    {\rule[-\textheight/2]{1ex}{\textheight}} 
  }{\textheight} %
}{0.5ex}\\           
#1\\                 
\rule{-1ex}{0ex}
\end{array}
}

\title[Spectral theory of TBG in magnetic fields]{Spectral theory of twisted bilayer graphene in a magnetic field}

\author{Simon Becker}
\address[Simon Becker]{ETH Zurich, Zurich, CH.}
\email{sion.becker@math.ethz.ch}

\author{Xiaowen Zhu}
\address[Xiaowen Zhu]{University of Washington, Seattle, USA.} 
\email{xiaowenz@uw.edu}

\begin{document}

\begin{abstract}
In this article we study the Bistritzer-MacDonald (BM) model with external magnetic field. We study the spectral properties of the Hamiltonian in an external magnetic field with a particular emphasis on the flat band of the chiral model at magic angles. Our analysis includes different types of interlayer tunneling potentials, the so-called \emph{chiral} and \emph{anti-chiral} limits. One novelty of our article is that we show that using a magnetic field one can discriminate between flat bands of different multiplicities, as they lead to different Chern numbers in the presence of magnetic fields, while for zero magnetic field their Chern numbers always coincide. 
\end{abstract}

\maketitle

\section{Introduction}
\label{s:intr}
When two graphene layers are stacked and twisted against each other, there exist specific angles coined the \emph{magic angles}, at which the composite material exhibits a form of superconductivity \cite{C18}. Twisted bilayer graphene is a promising platform to exhibit the integer and fractional quantum Hall effect (QHE) even without a magnetic field. While this has already been directly observed for the integer QHE \cite{Se20}, which is the so-called \emph{anomalous QHE}, an experimental verification of the anomalous FQHE is still missing. Motivated by this, the effect of small magnetic fields has been experimentally explored, in which case a version of the FQHE has been observed \cite{X21}. With our spectral analysis, following \cite{SS21}, we provide a mathematically rigorous foundation for this. Hamiltonians for the effective one particle band structure of twisted bilayer graphene have been derived in \cite{LPN07,BM11} where the authors observed that at specific magic angles, the Fermi velocity in the graphene sheets becomes zero. 

The magnetic BM model is, after a simple rescaling of the length scale (see our companion paper \cite{BKZ22} for an explanation), an effective one-particle continuum Hamiltonian 
\begin{equation}
    \label{eq:operator}
    H:= \begin{pmatrix} {H}^{B}_{D} &   {T}\\  {T}^* &   H^{B}_{D} \end{pmatrix},
\end{equation} 
where
\begin{equation}
    \label{eq: def_of_ab}
    H_D^B= \begin{pmatrix}
    0 & a^* \\ a & 0
    \end{pmatrix}\ \text{with~}\begin{cases}
    a = 2D_{\bar{z}} - A(z)\\
    a^* = 2D_z - \overline{A(z)} 
    \end{cases}, \  T(z) = \begin{pmatrix} \alpha_0 V(z) & \alpha_1 \overline{U_-}(z) \\  \alpha_1 U(z)& \alpha_0  V(z) \end{pmatrix}.
\end{equation}
Here, $H_D^{B}$ is the two-dimensional magnetic Dirac operator with magnetic vector potential $A$, effectively describing a single sheet of graphene in a magnetic field  close to zero energy, cf. \cite{W47}; $T$ represents the inter-layer tunnelling potential. The diagonal and off-diagonal terms in $T$, i.e. $\alpha_0 V_(z)$ and $\alpha_1 U(z)$, describe different types of inter-layer tunnellings, due to different stacking of atoms, (see Fig.\ref{fig: Moire}, \cite{BM11,BKZ22}). The limit of pure $\alpha_1$-coupling is called the \emph{chiral limit} and the limit of pure $\alpha_0$-coupling is called the \emph{anti-chiral limit}.
\begin{figure}[t!]
  \includegraphics[height=7cm]{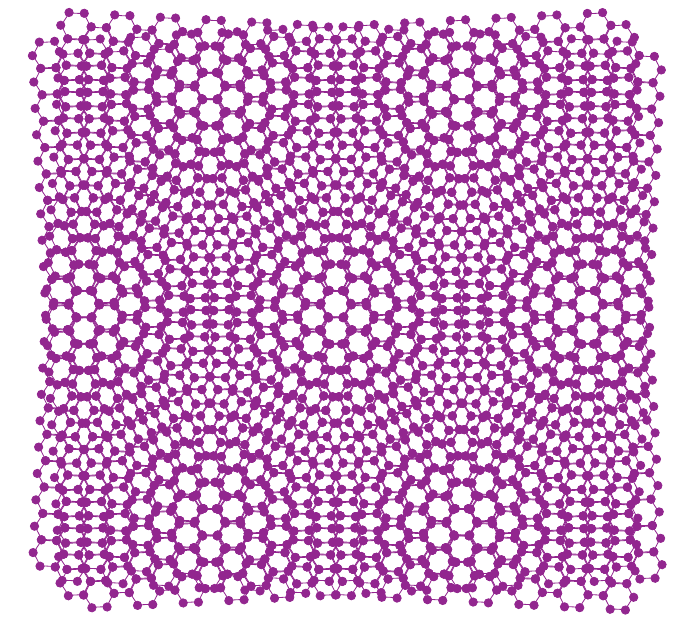}
  \includegraphics[height=7cm]{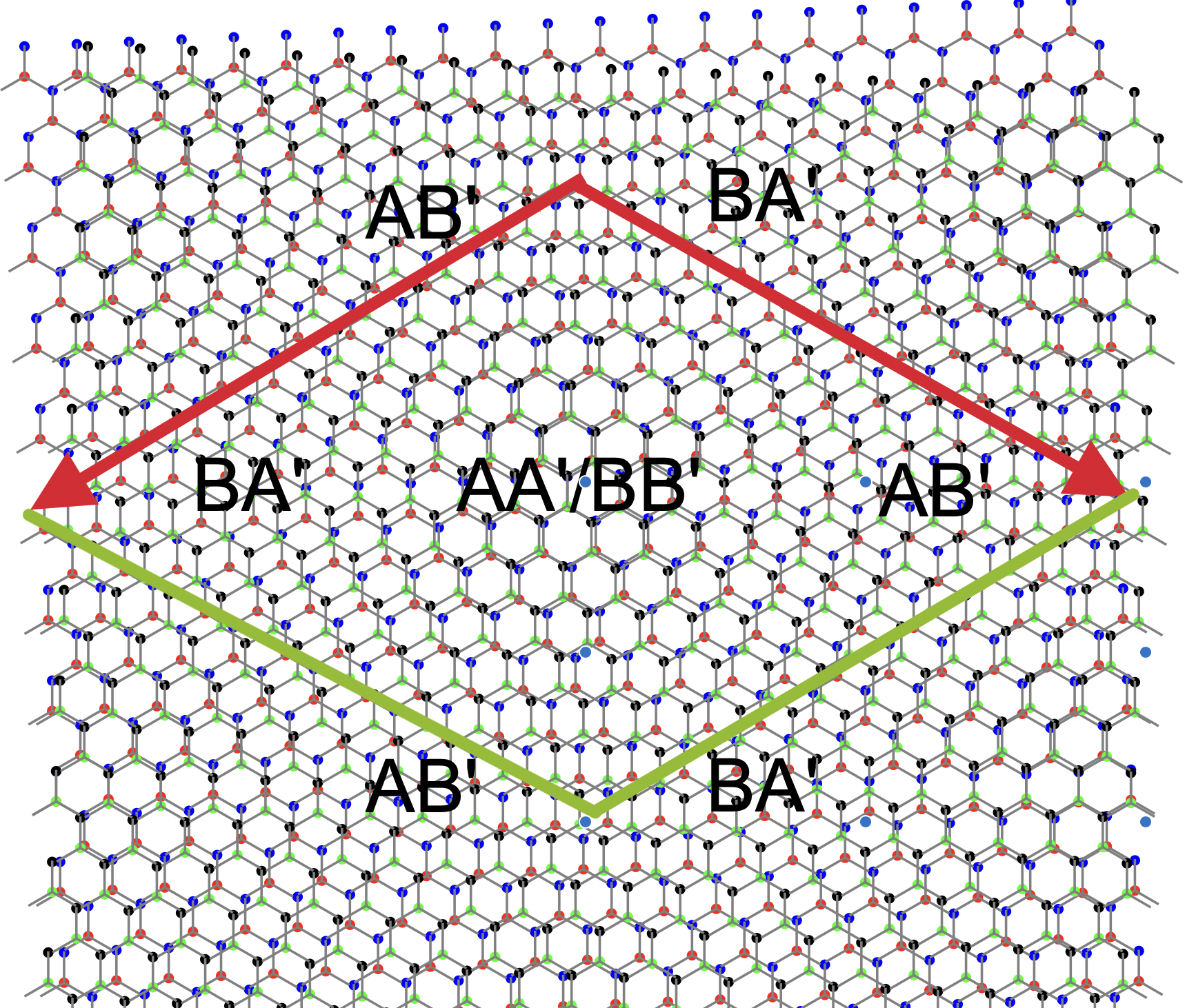}
\caption{Left: Visible moir\'e pattern at $\theta = 5^{\circ}$. Right: Single moir\'e hexagon, with (A={\color{red}red}, B={\color{blue}blue}) and (A'={\color{green}green}, B'={\color{black}black}) denote vertices of two sheets of graphene respectively.}
        \label{fig: Moire}
\end{figure}

In the BM model, the bands near zero energy appear almost flat, while it has been discovered in \cite{BEWZ20a,BEWZ20b,N21,NL22,TKV19} that the chiral limiting model exhibits perfectly flat bands at these so-called magic angles \cite{BEWZ20a,TKV19} whereas the anti-chiral model does not exhibit flat bands \cite{BEWZ20b}. In this article, we show that this dichotomy persists when a constant perpendicular magnetic field is applied \S\ref{sec:Proofsection}.

In this paper, we first perform a spectral and symmetry analysis of our model for various magnetic perturbations. This includes the existence and absence of perfectly flat bands at magic angles for different interlayer potentials and magnetic fields, see also \cite{BS99} for related results. This is shown in Section \ref{sec:Proofsection} in Theorem \ref{theo:magic_angles}.

In this row of mathematically rigorous results, we also want to mention the computer-assisted proof of the existence of a real magic angle by Luskin and Watson \cite{LW21} and other derivations of models for TBG \cite{CGG,Wa22,CM23}. 




Our article is then structured as follows:
\begin{itemize}
\item In Section \ref{sec:introduce_the_model}, we introduce the Bistritzer-MacDonald Hamiltonian for TBG with an additional external magnetic field.
\item In Section \ref{sec:floquet}, we introduce the magnetic Bloch-Floquet theory for the magnetic Bistritzer-MacDonald Hamiltonian.
\item In Section \ref{sec:chiral_lim}, we specialize to the chiral limit and study the existence and Chern number of flat bands at zero energy in constant magnetic fields. In particular, Theorem \ref{theo:away_mag_angl} focuses on the non-magic-angle case; Theorem \ref{theo:simple} and \ref{theo:double} focus on the simple and degenerate magic-angle case, respectively. 
\item In Section \ref{sec:Proofsection}, we prove that for the magnetic BM model,
\begin{itemize}
    \item periodic magnetic fields do not affect the presence of flat bands in Theorem \ref{theo:magic_angles}.
    \item flat bands persist under rational (with respect to the lattice) magnetic flux in Theorem  \ref{theo: eigenvalues}. 
\end{itemize}
\item In Section \ref{sec:loc_bands}, we prove that a plethora of almost flat bands is located exponentially close to zero energy in Theorem \ref{theo:exp_squeez_bands}. Finally, we conclude with some observations that show the presence of Landau-level type spectrum in the chiral model of twisted bilayer graphene in the absence of magnetic fields. 

\end{itemize}

\smallsection{Acknowledgements}
This research was partially supported by Simons 681675, NSF DMS-2052899 and DMS-2155211, DMS 2054589, and the Pacific Institute for the Mathematical Sciences. The contents of this work are solely the responsibility of the authors and do not necessarily represent the official views of PIMS. 

\section{Introduction of magnetic BM model} 
\label{sec:introduce_the_model}
We start by introducing relevant notation.

 \smallsection{Notation}
We commonly identify $\RR^2$ with $\CC$ by writing $x = (x_1,x_2) $ also as $ z=x_1+ix_2$. We write $ f = \mathcal O_\alpha ( g )_H $ if there is a constant $C_{\alpha}>0$ such that $ \| f \|_H \leq C_\alpha g $. We add the subscript $\CC_z$ or $\CC_\bfk$ to refer to complex plane of variable $z$ or $\bfk$ when it is ambiguous in the context. 

Given two vector space $A$ and $B$, we use $A\otimes B$ to denote the tensor product; $A\times B = \{(a,b):a\in A, b\in B\}$ to denote the direct product. Given a real (or complex) vector space $V$ and two subspaces $A, B \subset V$, we use $A\oplus B = \{c_1a + c_2b: c_1, c_2\in \RR \text{~(or~}\CC), a \in A,b\in B \}$ to denote the inner direct sum. Denote $0_{\CC^2} = (0,0)^T\in \CC^2$.

\subsection{Moir\'e lattices and TBG}
\label{subsec: BM_model}

When two honeycomb lattices are stacked on top of each other and twisted by an angle $\theta$, a moir\'e honeycomb pattern becomes visible, cf. Fig. \ref{fig: Moire}. This new honeycomb structure has length scale $\lambda_{\theta}.$\footnote{In fact, $\lambda_\theta = \frac{C\sqrt{3}}{2\sin (|\theta|/2)}$ when $|\theta|<\pi/6$ by \cite{RK93,LPN07}.} 


To simplify the discussion, we start by introducing a unit-size honeycomb lattices, i.e. the one with side length $\frac{4\pi}{\sqrt{3}}$. 

Let $\omega = \exp(\frac{2\pi i}{3})$,  $\zeta_1 = \frac{4\pi i \omega}{3}$ and $\zeta_2 = \frac{4\pi i \omega^2}{3}$. A \textit{unit-size honeycomb lattice} is then invariant under translations by the triangular lattice $\Gamma = \zeta_1\ZZ \oplus \zeta_2\ZZ$. The unit cell, its dual lattice, and the unit cell of the dual lattice (Brillouin zone) are denoted by $E = \CC/\Gamma$, $\Gamma^* = \eta_1\ZZ \oplus \eta_2\ZZ$, and $E^* = \CC/\Gamma^*$, respectively, where $\eta_1 = \sqrt{3}\omega^2$ and $\eta_2 = -
\sqrt{3}\omega$. Since the moir\'e honeycomb lattice changes its length scale depending on the twisting angle, we also define a \textit{rescaled honeycomb lattice} with scaling parameter $\lambda \in \mathbb Z^2$ by $\Gamma_\lambda = \lambda_1 \zeta_1\ZZ \oplus \lambda_2\zeta_2\ZZ$, $E_\lambda = \CC/\Gamma_\lambda$, $\Gamma_\lambda^* = \lambda_1^{-1}\eta_1\ZZ \oplus \lambda_2^{-1}\eta_2\ZZ$ and $E_\lambda^* = \CC/\Gamma_\lambda^*$. Finally, we denote by $T_{1,\lambda},T_{2,\lambda} \in \mathcal L(L^2(\mathbb C))$ the standard translation operators $(T_{j,\lambda}u)(z):=u(z+\lambda_j\zeta_j)$ associated with the lattice $\Gamma_{\lambda}$. In particular, when $\lambda_j = 1$, we denote $(T_ju)(z):= u(z + \zeta_j)$. 

\subsection{Chiral and anti-chiral limits} 
The two basic examples of tunneling potentials, $U$ and $V$, are defined as 
\[
        U(z) = \sum_{\ell = 0}^2 \omega^{\ell} \exp\left( \frac{z \bar{\omega}^{\ell} - \bar{z} \omega^{\ell}}{2}\right),\quad
        V(z) = \sum_{\ell = 0}^2 \exp\left( \frac{z\bar{\omega}^{\ell} - \bar{z}\omega^{\ell}}{2}\right).
\]
For $\bfa = n_1 \zeta_1 + n_2 \zeta_2$, the potentials are defined to satisfy the following symmetries
\begin{equation}
  \label{eq:symmU}
    \begin{split}
      &V ( z + \bfa ) = \bar \omega^{n_1 + n_2} V ( z ) ,\ \ V ( \omega z ) = V ( z ) , \ \  V(\bar{z}) = \overline{V(z)} = V(-z), \\
      &U ( z + \mathbf a ) = \bar \omega^{n_1 +n_2} U ( z ) ,  \ \ U ( \omega z ) = \omega U ( z ) , \ \ U ( \bar z ) = \overline{ U ( z )}, \ \ U_-(z):=U(-z).
    \end{split}
\end{equation}

\begin{figure}[t!]
    \begin{subfigure}[b]{0.32\textwidth}
        \centering
        \includegraphics[width=7cm, height=1.8in]{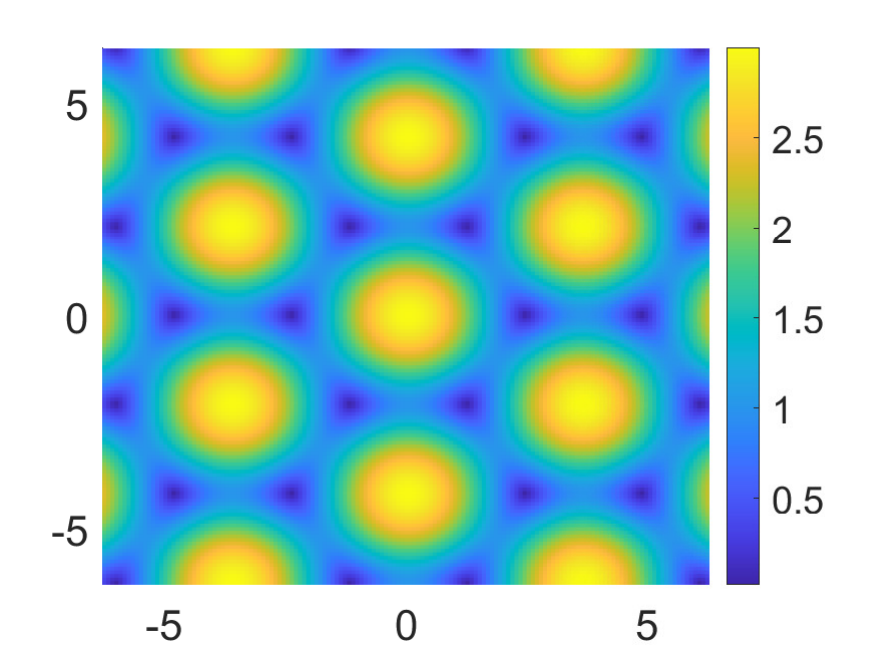}
        \caption{tunneling potential $\vert V \vert^2$ for $\mathrm{AA^{\prime}}/\mathrm{BB^{\prime}}$-coupling. }
       \label{fig:hcombAA}
    \end{subfigure}
        \begin{subfigure}[b]{0.32\textwidth}
        \centering
        \includegraphics[width=6.5cm,height=1.8in]{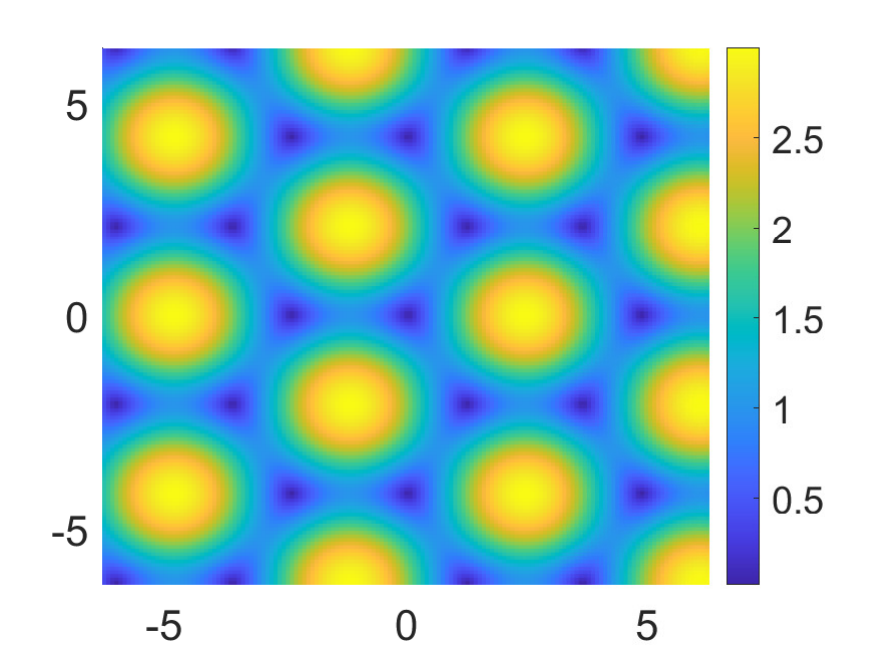}
        \caption{tunneling potential $\vert U \vert^2$ for $\mathrm{AB'}$-coupling.}
        \label{fig:fcellBA}
    \end{subfigure}
     \begin{subfigure}[b]{0.32\textwidth}
        \centering
        \includegraphics[width=6 cm,height=1.8in]{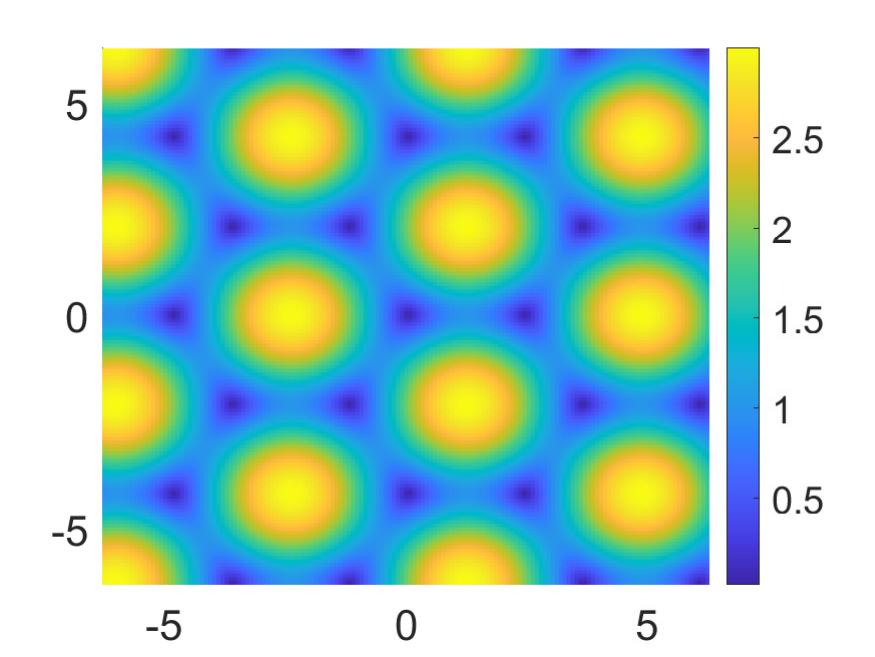}
        \caption{tunneling potential $\vert U_- \vert^2$ for $\mathrm{BA'}$-coupling.}
        \label{fig:fcellAB}
    \end{subfigure}%
        \caption{The tunneling potentials for different coupling types on unit-size honeycomb lattice.}\label{fig:coupling}
\end{figure}

\subsection{Spectral properties of chiral and anti-chiral model}
We start by introducing two different limiting cases of the BM Hamiltonian: 

The \emph{chiral model} is the Hamiltonian \eqref{eq:operator} when $\alpha_0=0$. When conjugated by the unitary $\mathscr U=\operatorname{diag}(1,\sigma_1,1)$, $H_{\operatorname{c}}(\alpha_1, B):=\mathscr U  H\mathscr U$ has off-diagonal form
\begin{equation}
\begin{split}
\label{eq:DcB}
H_{\operatorname{c}}(\alpha_1,B)&=\begin{pmatrix}0 & (D_c(\alpha_1,B))^* \\ D_c(\alpha_1,B) & 0  \end{pmatrix} \text{ with } 
 D_c(\alpha_1,B)= \begin{pmatrix}  a  & \alpha_1 U(z) \\ \alpha_1 U_-(z) &  a  \end{pmatrix},
\end{split}
\end{equation}
where $a=2 D_{\bar z}-A(z).$

The other limiting model is the \emph{anti-chiral model} which is obtained by setting $\alpha_1=0$. It can also be transformed into an off-diagonal form when conjugated by a unitary $\mathcal V$, giving rise to a Hamiltonian
\begin{equation}
\begin{split}
\label{eq:achiral}
&H_{\ach}(\alpha_0,B):=\mathcal V H\mathcal V=\begin{pmatrix} 0& (D^{\theta}_{\ach}(\alpha_0,B))^* \\ D^{\theta}_{\ach}(\alpha_0,B)&0 \end{pmatrix} \text{ with } \\
&D_{\ach}^{\theta}(\alpha_0,B) = \begin{pmatrix} \alpha_0 V(z) & e^{i \theta/2}a   \\  e^{i \theta/2}a^* & \alpha_0 \overline{V(z)} \end{pmatrix}\text{ where }\\
&\mathcal V = \begin{pmatrix} \mathcal V_1 & \mathcal V_2 \\\mathcal V_2 & \mathcal V_1\end{pmatrix} \text{ for }\mathcal V_1 = \operatorname{diag}( i\lambda , 0 ), \ \  \mathcal V_2 = \operatorname{diag}(0, -\bar{\lambda}), \ \ \lambda=e^{i \frac{\pi}{4}}. 
\end{split}
\end{equation}
Both the chiral and anti-chiral model can be cast in an off-diagonal matrix form. This implies that for both the chiral and anti-chiral model with magnetic field, the spectrum is symmetric with respect to zero. To see this directly, let $U:=(\sigma_3 \otimes \operatorname{id}_{\CC^2})$,  then by directly conjugating the Hamiltonians, we find $UH_cU = - H_c $ and $U H_{\ach}^{\theta}U =-H_{\ach}^{\theta}.$

\subsection{Characterization of magic angles}\label{subsec: charofmagicangles} It has been shown \cite{BEWZ20a} that in the absence of magnetic fields, i.e. $A=0$, there exists a discrete set $\mathcal A \subset \mathbb C$, the set of magic angles, such that for $\alpha_1 \in \mathcal A$, the chiral model $H_c(\alpha_1, B = 0)$ \eqref{eq:DcB} exhibits flat bands at energy zero. 

More explicitly, this set $\mathcal A$ has been characterized by the eigenvalues of a compact operator 
\[\begin{split} &\alpha_1 \in \mathcal A \Leftrightarrow \alpha_1^{-1}\in \Spec(T_{\mathbf k})\text{ for some } \mathbf k \notin 3\Gamma^* \\
&T_{\mathbf k} = (2D_{\bar z}-\mathbf k)^{-1} \begin{pmatrix} 0 & U(z)\\U(-z) & 0 \end{pmatrix}: L^2_0(\CC/\Gamma) \to L^2_0(\CC/\Gamma).
\end{split}\]
The space $L^2_0(\CC/\Gamma)$ is defined as 
\[ L^2_0(\CC/\Gamma):=\{u \in L^2(\CC/\Gamma); L_{a} u = u \text{ with }a \in \tfrac{\Gamma}{3} \}\]
where $L_a u(z) = \operatorname{diag}(\omega^{a_1+a_2},1) u(z+a)$ with $a = \frac{4}{3}\pi i (\omega a_1+\omega^2 a_2)$ for $a_j \in \ZZ.$


\section{Magnetic Bloch-Floquet theory}
\label{sec:floquet}
In this section, we introduce the magnetic Bloch-Floquet theory following \cite{Sj89} in preparation of studying the flat bands.

Bloch-Floquet theory yields a decomposition of an operator $S$ with translational symmetries with respect to a lattice, say $\Gamma$, into a family of simpler operators $\{S_\bfk\}_{\bfk\in\CC/\Gamma^*}$. To obtain such a decomposition, we first need to find an Abelian group of translations $\{T_{\bfn}\}_{\bfn\in\ZZ^2}$ on $\Gamma$ that commute with $S$. Using these translations, one can construct a unitary operator, called the Gelfand transform $U$, such that $USU^{-1} = \int^\oplus_{\CC/\Gamma^*} S_\bfk d\bfk$.

\begin{assumption}[Magnetic potential]
\label{assump: magnetic potential}
We assume that the magnetic potential is of the form $A = A_{\operatorname{con}} + A_{\operatorname{per}}$, where
\begin{enumerate}
    \item $A_{\operatorname{con}}(z)=\frac{Bi}{2}z$ generates a constant magnetic field $B$,
    \item $A_{\operatorname{per}}$ is real-analytic and periodic w.r.t. $\Gamma_\lambda = \lambda_1\zeta_1\ZZ\oplus \lambda_2\zeta_2\ZZ$, for some $\lambda = (\lambda_1,\lambda_2)\in \ZZ^2$,
\end{enumerate}
such that $\Phi_{\operatorname{mag}}:=B|\CC/\Gamma_\lambda| \in 2\pi \ZZ$. In particular, we denote 
\begin{equation}
    \label{eq: assumption}
v_j := \lambda_j\zeta_j, \  \Gamma_{\operatorname{mag}} := \Gamma_\lambda, \  p:=\frac{\Phi_{\operatorname{mag}}}{2\pi}, \ q:=\lambda_1\lambda_2 , \text{~and~} \Phi:= B|\CC/\Gamma| = 2\pi \tfrac{p}{q}\in 2\pi \mathbb{Q},
\end{equation}
for some $p,q\in \ZZ$.
The integer $p$ is called the number of \emph{Dirac flux quanta}.
\end{assumption}

We start by defining magnetic translations on $\Gamma_{\operatorname{mag}} = v_1\ZZ\oplus v_2\ZZ$. To this end, recall that the standard translation by $v_j$ have been denoted by $T_{j,\lambda}$ in Sec \ref{subsec: BM_model}. Then the basic magnetic translation by $v_j$ can be defined as
\begin{equation}
    \label{eq: def_of_magnetic_translation}
    T_{j,\text{mag}} := e^{i\varphi_j} T_{j,\lambda} \text{~with~} \varphi_j(z) = -\Re(\overline{A_{\operatorname{con}}(v_j)} z), \quad j = 1,2,
\end{equation}
which satisfy $2D_z(i\varphi_j) = -\overline{A_{\operatorname{con}}(v_j)}$ and  $2D_{\bar z}(i\varphi_j)=- A_{\operatorname{con}}(v_j)$ such that 
\begin{equation}
\label{eq: commute1}
[a^*,T_{j,\text{mag}}] = [a,T_{j,\text{mag}}] = 0, \quad \text{where~} a = 2D_{\bar z} - A(z).
\end{equation}
Furthermore, for our model, $A_{\operatorname{con}}(z) = \frac{Biz}{2}$ implies $\varphi_j(z)  = \frac{Bi}{4}(\overline{v_j}z - v_j\bar{z})$ is linear in $z$ and $\bar z$. Thus we can see
\begin{equation} \begin{split}
\label{eq: magtrans_commute}
T_{1,{\operatorname{mag}}}T_{2,{\operatorname{mag}}} &=e^{i\varphi_1+i\varphi_2(\bullet+ v_1)}T_{1,\lambda}T_{2,\lambda} = e^{i(\varphi_2( v_1) -\varphi_1( v_2))}e^{i\varphi_2 + i\varphi_1(\bullet +  v_2)}T_{2,\lambda}T_{1,\lambda}\\
&= e^{i (\varphi_2( v_1)-\varphi_1( v_2))}T_{2,{\operatorname{mag}}}T_{1,{\operatorname{mag}}}\\
& = e^{-\frac{B}{2}(\overline{v_1}v_2 - \overline{v_2}v_1)}T_{2,\text{mag}}T_{1,\text{mag}}\\
&=e^{-iB|\CC/\Gamma_{\operatorname{mag}}|}T_{2,{\operatorname{mag}}}T_{1,{\operatorname{mag}}} = e^{-2\pi p i}T_{2,{\operatorname{mag}}}T_{1,{\operatorname{mag}}}=T_{2,{\operatorname{mag}}}T_{1,{\operatorname{mag}}}.
\end{split}
\end{equation}
where we used $\overline{v_1}v_2 - \overline{v_2}v_1 = 2i|\CC/\Gamma_{\operatorname{mag}}|$.

The commutativity of translations above allows us then to define general magnetic translation on $\Gamma_{\operatorname{mag}}$:
\begin{equation}
    \label{eq: def_of_Tn}
    T_{\textbf{n},\text{mag}} := T_{1,\text{mag}}^{n_1}T_{2,\text{mag}}^{n_2}\text{~for~}\textbf{n} = (n_1,n_2)\in\ZZ^2.
\end{equation}
By \eqref{eq: commute1} and \eqref{eq: magtrans_commute}, we have for $\textbf{n},\textbf{n}'\in \ZZ^2$, 
\begin{equation}
    \label{eq:commutator}
    \begin{split}
    &T_{\textbf{n},\text{mag}}T_{\textbf{n}',\text{mag}} = T_{\textbf{n} + \textbf{n}',\text{mag}} = T_{\textbf{n}',\text{mag}}T_{\textbf{n},\text{mag}}, \text{~and~}\\
    &[T_{\bfn, \text{mag}}, a] = [T_{\bfn, \text{mag}}, a^*] = 0.
    \end{split}
\end{equation}
Thus, $\{T_{\textbf{n},\text{mag}}\}_{\textbf n\in\ZZ^2}$ forms an abelian group of translations on $\Gamma_{\operatorname{mag}}$ that commutes with $a$ and $a^*$. 

Furthermore, $T_{\bfn, \text{mag}}$ gives rises to Abelian groups that commute with many more complicated operators. We summarize them all below: 
\begin{lemm}
\label{lemm: Abelian_group}
Let $A$ satisfy Assumption \ref{assump: magnetic potential}. Let $T_{\bfn,\operatorname{mag}}$ be defined as in \eqref{eq: def_of_Tn}. Then the following families of translations form Abelian groups and
\begin{enumerate}
    \item $\{T_{\bfn,\operatorname{mag}}:\bfn\in\ZZ^2\}$ commutes with $a,a^*$ in \eqref{eq:operator},
    \item $\{\diag(1, 1)\otimes T_{\bfn,\operatorname{mag}}:\bfn\in\ZZ^2\}$ commutes with $H_D^B$ in \eqref{eq:operator},
    \item $\{\diag(\omega^{\bfn\cdot\lambda},\omega^{\bfn\cdot\lambda},1,1) \otimes T_{\bfn,\operatorname{mag}}:\bfn\in\ZZ^2\}$ commutes with $H$ in \eqref{eq:operator},
    \item $\{\diag(\omega^{\bfn\cdot\lambda},1)\otimes T_{\bfn,\operatorname{mag}}:\bfn\in\ZZ^2\}$ commutes with $ D_{\text{c}}(\alpha_1, B)$ in \eqref{eq:DcB},
    \item $\{\diag(\omega^{\bfn\cdot\lambda},1,\omega^{\bfn\cdot\lambda},1)\otimes T_{\bfn,\operatorname{mag}}:\bfn\in\ZZ^2\}$ commutes with $H_{\text{c}}(\alpha_1, B)$ in \eqref{eq:DcB}, 
    \item $\{\diag(-\omega^{\bfn\cdot\lambda},-1,1,-\omega^{\bfn\cdot\lambda})\otimes T_{\bfn,\operatorname{mag}}:\bfn\in\ZZ^2\}$ commutes with $H_{\text{ac}}(\alpha_0, B)$ in \eqref{eq:DcB},
\end{enumerate}
where $\bfn\cdot\lambda := n_1\lambda_1 + n_2 \lambda_2$, $\omega = \exp(\tfrac{2\pi i}{3})$. 
\end{lemm}
\begin{proof}
$(1)$ has already been established in \eqref{eq: commute1}. For the remaining parts, by \eqref{eq:symmU} and \eqref{eq: def_of_magnetic_translation}, we have
\begin{equation}
\begin{cases}
    XT_{\bfn, \text{mag}} = \omega^{\bfn\cdot\lambda}T_{\bfn, \text{mag}}X, &\text{~for~}X\in \{U(z),V(z),\overline{U_-}(z)\},\\
    YT_{\bfn, \text{mag}} = \omega^{\bfn\cdot\lambda}T_{\bfn, \text{mag}}Y,&\text{~for~}Y\in\{\overline{U}(z),\overline{V}(z), U_-(z)\}.
\end{cases}
\end{equation}
Then the claims follow by definition and direct computation.
\end{proof}
In particular, throughout the article, we will abuse the notation in the following way: When the operator is clearly specified, we will use the same notation $\mathscr L_\bfn$ to denote the translations above for the corresponding operators. For example, the $\mathscr L_\bfn$ for $H_c(\alpha_1, B)$ will be $\mathscr L_\bfn = \diag(\omega^{\bfn\cdot\lambda},1,\omega^{\bfn\cdot\lambda},1)\otimes T_{\bfn,\operatorname{mag}}$. 

Furthermore, we will build up the general Bloch-Floquet theory using this kind of (abused) general notation. The $\mathcal H_f$, $\mathcal H_\tau$, $\mathcal U_B$ below all depends on the operator but will be abused throughout the paper when the corresponding operator is clear.  
\begin{theo}[Magnetic Bloch-Floquet theory]
\label{theo: Floquet theory}
Assume $\{\mathscr L_\bfn\}_{\bfn\in\ZZ^2}$ form a family of translations, as specified in  Lemma \ref{lemm: Abelian_group}, commuting with some operator $S=S(z,\bar z,D_{z}, D_{\bar z})$ on $L^2(\CC;\CC^d)$. Then, we define the Hilbert space
\begin{equation}
\label{eq:Hf}
    \mathcal H_f:= \{ f \in L^2_{\operatorname{loc}}(\CC;\CC^d); \mathscr L_{\mathbf n} f(z)= f(z) \text{ for all }\textbf n\in \mathbb Z^2\}
\end{equation}
which naturally defines a Hilbert space with inner product $\langle f,g \rangle_{\mathcal H_f}:=\int_{\CC/\Gamma_{\operatorname{mag}}} \langle f(y), g(y)\rangle_{\CC^d} \ dy$. We then define the Hilbert space
\[\mathcal H_{\tau}:=\{ f \in L^2_{\operatorname{loc}}(\CC; \mathcal H_f); f(\mathbf k-\gamma^*,z) = e^{-i\Re(z\bar{\gamma}^*)} f(\mathbf k,z) \text{ for all }\gamma^* \in \Gamma^*_{\operatorname{mag}}\}.\]
The Gelfand transform $\mathcal U_B: L^2(\CC) \to \mathcal H_{\tau}$ is defined for $u \in C_c(\CC)$ by
\[ \mathcal U_B u(\mathbf k,z) := \sum_{\mathbf n \in \ZZ^2} e^{-i \Re((n_1v_1+n_2 v_2-z)\overline{ \mathbf k})} (\mathscr L_{\mathbf n} u)(z)\]
and extends to a unitary operator with inverse
\[ \mathcal U_B^{-1} \phi(z) = \int_{\CC/\Gamma^*_{\operatorname{mag}}} e^{-i\Re(z\overline{\mathbf k})} \phi(\mathbf k,z) \frac{d\mathbf k}{\vert \CC/\Gamma^*_{\operatorname{mag}}\vert}, \]
such that for fiber operators $S_{\mathbf k}(z,\bar z,D_{z},D_{\bar z}):=S(z,\bar z,D_{z}+\tfrac{\overline{\mathbf k}}{2},D_{\bar z}+\tfrac{\mathbf k}{2})$ on $\mathcal H_{f}$
\[ \mathcal U_B S \ \mathcal U_B^* = \int^{\oplus}_{\CC/\Gamma_{\operatorname{mag}}^*} S_{\mathbf k}(z,\bar z,D_{z},D_{\bar z}) \ d\mathbf k.\]
\end{theo}
\begin{rem}
The magnetic Sobolev spaces associated with \eqref{eq:Hf}, when working with operators mentioned in Lemma \ref{lemm: Abelian_group}, are defined for $m \in \mathbb N$ as
\[ H^m_{\mathcal H_f}:=\{f \in \mathcal H_f; (a^{*})^{\alpha_1}a^{\alpha_2} f \in \mathcal H_f, \vert \alpha \vert \le m\}, \]
with inner product
\[ \langle f,g\rangle_{H^m_{\mathcal H_f}}:=\sum_{\vert \alpha \vert \le m} \langle (a^{*})^{\alpha_1}a^{\alpha_2} f,(a^{*})^{\alpha_1}a^{\alpha_2} g\rangle_{\mathcal H_f}.\]
\end{rem}

We first consider the case when $A_{\operatorname{per}} = 0$. Notice that the fiber operators of $a = 2D_{\bar{z}} - A(z)$ are $a_\bfk = a + \bfk$. Furthermore, for $A(z) = \frac{Bi}{2}z$, we have $[a_\bfk,a_\bfk^*] = 2B.$ This implies that $\ker_{\mathcal H_f}(a_\bfk) = \{0\}$ for $B<0$ and $\ker_{\mathcal H_f}(a_\bfk^*) = \{0\}$ for $B>0.$ We shall next construct $\operatorname{ker}_{\mathcal H_f}(a_\bfk)$ for $B>0.$ 

Before we start, let us recall the nullspace of $a_\bfk$ on $L^2(\CC)$. Due to the conjugacy relation $a_\bfk = e^{-\frac{i\bfk \bar{z}}{2}}e^{-B \vert z \vert^2/4}(2D_{\bar z})e^{B \vert z \vert^2/4} e^{\frac{ik\bar{z}}{2}}$, this is precisely the (transformed) Bargmann space
\begin{equation}
\label{eq:nullsp}
    \text{ker}_{L^2(\CC)}(a_\bfk) = \left\{e^{-\frac{i\bfk \bar{z}}{2}} e^{-B|z|^2/4}f(z); f \in \mathscr O(\CC), \int_{\CC} \vert f(z) \vert^2 e^{-B\vert z \vert^2/2} \ dz <\infty\right\},
 \end{equation}
where $\mathscr O(\CC)$ is the set of entire functions.

While the characterization of the nullspace of $a_\bfk$ , and thus of $H_D^B$, on $L^2(\CC)$ is straightforward and perhaps well-known as above, we will need the nullspace of $a_{\mathbf k}$ and $H_{D,\mathbf k}^B$ on the magnetic torus $\CC/\Gamma_{\operatorname{mag}}$ associated with the honeycomb lattice, instead. In another words, we need $\ker_{\mathcal H_f}(a_\bfk)$ rather than $\ker_{L^2(\CC)}(a_\bfk)$. As we will see, the ansatz of finding $\ker_{\mathcal H_f}(a_\bfk)$ is inspired by the structure of $\ker_{L^2(\CC)}(a_\bfk)$.
\begin{theo}\label{theo:Fredholm}
Assume $A(z)$ satisfy Assumption \ref{assump: magnetic potential} with $A_{\operatorname{per}}=0$ and $B>0$, we have  $\ker(a_\bfk^*) = \{0\}$ while $\ker_{\mathcal H_f}(a_{\mathbf k})$ is composed of functions
\begin{equation}
    \label{eq: function_u}
u_{\mathbf k}(z) = e^{-\frac{i\mathbf k \bar z}{2}} e^{B \frac{(z^2 - \vert z \vert^2)}{4}} g(z)
\end{equation} 
with $g(z) = e^{\gamma z} \prod\limits_{i=1}^p \tilde \sigma(z-z_i)$ for arbitrary $\gamma\in \CC$ and $z_i\in \CC$ satisfying the following two constraint equations
\begin{equation}
\label{eq:condition}
\gamma v_j - i\frac{\mathbf k \overline{v_j}}{2}  + ip\pi = \xi_j \sum_{j=1}^p z_j, \qquad  j \in \{1,2\}.
\end{equation}
Here $v_j:=\lambda_j \zeta_j$, $\xi_j :=  \frac{2\pi i \lambda_j B }{3}$ and $\tilde \sigma$ is the modular-invariant Weierstra{\ss} function, as defined in \cite{H18}. In particular, $\dim\ker_{\mathcal H_f}(a_\bfk) = p$. 

Finally, $a_{\mathbf k}$ is a Fredholm operator of index
\[
\operatorname{ind}(a_{\mathbf k}) = \operatorname{dim}\ker_{\mathcal H_f}(a_\mathbf k)-\operatorname{dim}\ker_{\mathcal H_f}(a^*_{{\mathbf k}})=p
\]
and $\ker_{\mathcal H_f \times \mathcal H_f}(H_{D,\mathbf k}^B) =\ker_{\mathcal H_f}(a_\mathbf k) \times \{0\}.$
\end{theo}
\begin{proof}
We start with the ansatz $u_{\mathbf k}(z) = e^{-\frac{i\mathbf k \bar z}{2}} e^{B\frac{(z^2 - \vert z \vert^2)}{4}} g(z)$, inspired by the nullspace of $a$ on $L^2(\CC)$ in \eqref{eq:nullsp}, where $g$ is an entire function that we will determine soon.

Recall that $v_j = \lambda_j\zeta_j$, and our ansatz $u_\bfk$ satisfies $T_{j,\text{mag}}u_\bfk(z) = u_\bfk(z)$, i.e. 
\[
\begin{split}
    1 &= \frac{T_{j,\text{mag}}u_k}{u_k} = \frac{e^{i\varphi_j(z)}u_\bfk(z + v_j)}{u_\bfk(z)}\\
    &= \frac{g(z + v_j)}{g(z)} \exp\left(\tfrac{B}{4}(v_j\bar{z} - \overline{v_j} z) - \tfrac{i\bfk \overline{v_j}}{2} + \tfrac{B}{4}(v_j^2 + 2v_jz - |v_j|^2 - z\overline{v_j} - \overline{z}v_j)\right)\\
    &= \frac{g(z + v_j)}{g(z)}\exp\left(-\tfrac{i\bfk \overline{v_j}}{2} + \tfrac{Bz}{2}(\overline{v_j} - v_j) +\tfrac{Bv_j}{4}(\overline{v_j} - v_j)\right).
\end{split}
\]
Thus we get the boundary constraints of $g(z)$ on the unit cell $\CC/\Gamma_{\text{mag}}$:
\begin{equation}
        \label{eq: bdry_cond}
    g(z + v_j) = g(z)\exp\left(\tfrac{i\bfk\overline{v_j}}{2} + \xi_j z + S_j\right),
\end{equation}
where $\xi_j  = \tfrac{B}{2}(\overline{v_j} - v_j) = \tfrac{2\pi i B\lambda_j}{3}$ and $S_j = \tfrac{\xi_j v_j}{2}$. 

To find entire functions $g(z)$ satisfying \eqref{eq: bdry_cond}, we first count the number of zeros of $g$ by the argument principle. By $\overline{v_1}v_2 - \overline{v_2}v_1 = 2i|\CC/\Gamma_{\text{mag}}|$ and \eqref{eq: assumption}, we have
\begin{equation}
\label{eq:Legendre}
 \begin{split}
\frac{1}{2\pi i}\int_{\partial (\CC/\Gamma_{\operatorname{mag}})} \frac{g'(z)}{g(z)} \ dz &=\frac{1}{2\pi i} \left(\int_{\frac{-v_1+ v_2}{2}}^{-\frac{v_1+ v_2}{2}} d\left(\log \tfrac{g(z)}{g(z+v_1)} \right) + \int_{-\frac{v_1+ v_2}{2}}^{\frac{v_1- v_2}{2}} d\left(\log \tfrac{g(z)}{g(z+v_2)} \right) \right) \\
&= \frac{v_2 \xi_1 - v_1 \xi_2}{2\pi i} = \frac{B(\overline{v_1}v_2 - \overline{v_2}v_1)}{4\pi i} = \frac{2Bi|\CC/\Gamma_{\text{mag}}|}{4\pi i}  = p.
\end{split}
\end{equation}
Hence $g$ has $p$ zeros (counting multiplicity) in a unit cell $\CC/\Gamma_{\operatorname{mag}}$. We denote them by $z_1, \dots, z_p$. We also denote the non-repeated zeros of $g$ by $w_1, \dots, w_{p_0}$ with multiplicity $n_1,\dots, n_{p_0}$, where $n_1 + \dots + n_{p_0} = p$. 

To determine $g$, we further notice that, by the boundary constraints \eqref{eq: bdry_cond},
\[
    \frac{g'(z+v_j)}{g(z+v_j)} =\frac{g'(z)}{g(z)} +\xi_j, \qquad j = 1,2.
\]
As a result, $f:=(g'/g)'$ is a meromorphic, double periodic\footnote{i.e. $f(z + v_j) = f(z)$, $j = 1,2$.} function; thus, an elliptic function. Furthermore, $f$ has (and only has) a pole of order $2$ at each (non-repeated) zero of $g$, i.e. $w_i$, $i = 1,\dots, p_0$. Such kind of elliptic function with poles of order $2$ can be constructed using the so-called Weierstra{\ss} $\wp$-function 
\[
        \wp(z) = \frac{1}{z^2} + \sum\limits_{\gamma\in\Gamma_{\text{mag}}^{\times}}\left(\frac{1}{(z - \gamma)^2} - \frac{1}{\gamma^2}\right), \quad \text{~where~} \Gamma_{\text{mag}}^\times = \Gamma_{\text{mag}} \setminus \{0\}. 
\]
The Weierstra{\ss} $\wp$-function is double periodic w.r.t. $\Gamma_{\text{mag}}$, holomorphic on $\CC\setminus \Gamma_{\text{mag}}$, and has poles of order $2$ at lattice points $\Gamma_{\text{mag}}$. From $\wp$, we can construct $f$ in the following way: 

Denote the Laurent expansion of $f$ near each $w_i$ as $f(z) = \sum\limits_{k = -2}^{+\infty} c_{k}^{(i)} (z - w_i)^k$. Since $w_i$ is a zero of order $n_i$ for $g$, we see that 
\[
c_{-2}^{(i)} = \lim\limits_{z\to w_i} f(z)(z - w_i)^2 = \lim\limits_{z\to w_i} \left(\frac{g'(z)}{g(z)}\right)'(z -w_i)^2 = -n_i.
\]
Then $F(z) := f(z) + \sum\limits_{i = 1}^{p_0} n_i \wp(z - w_i)$ has no poles anymore because $f$ and $\wp$ only have poles of order $2$ but they are all eliminated by the summation. As a result, $F(z)$ is entire, double periodic; thus, by Liouville's theorem, a constant. We denote it by $c_0$. Hence we have
\begin{equation}
    \label{eq: log_g(z)}
 \left(\frac{d}{dz}\log(g(z))\right)' = \left(\frac{g'(z)}{g(z)}\right)' = f(z) = -\sum_{i = 1}^p n_i \wp(z - w_i) + c_0.
\end{equation}
To solve for $g$, recall that the Weierstra{\ss} sigma function (See Appendix \ref{sec:Haldane} for more details)
\[
\sigma(z) = z \prod_{\nu \in \Gamma_{\operatorname{mag}}^\times} \Big(1-\frac{z}{\nu} \Big)e^{\frac{z}{\nu}+\frac{z^2}{2\nu^2}}\quad \text{~satisfies~} \quad \left(\frac{d}{dz}\log(\sigma(z))\right)' = \left(\frac{\sigma'(z)}{\sigma(z)}\right)' = -\wp(z). 
\]
Plug into \eqref{eq: log_g(z)} and integrate \eqref{eq: log_g(z)} twice, we get 
\[
\begin{split}
   & \log (g(z)) = \sum\limits_{i = 1}^{p_0} n_i \log (\sigma(z - w_i)) + c_0z^2 + c_1z + c_2,\\
   \Rightarrow \quad & g(z) = e^{c_0z^2 + c_1z + c_2} \prod_{i = 1}^{p_0} \left(\sigma(z - w_i)\right)^{n_i} =: ce^{p(z)}\prod_{i = 1}^{p_0} \left(\sigma(z - w_i)\right)^{n_i}.
\end{split}
\]
To solve for $p$, we first recall the ``modified'' Weierstra{\ss} sigma function $\tilde \sigma(z)$, introduced by Haldane in \cite[eq. (11)]{H18}, which has the following properties: 
\begin{equation}
    \label{eq: sigma_bdry}
\tilde \sigma(z) = e^{-\frac{1}{2}\gamma_2(\Gamma_{\text{mag}})z^2}\sigma(z), \quad \text{~with~}\quad 
\tilde \sigma(z+v_j) = -e^{\frac{\xi_j z + S_j}{p}} \tilde \sigma(z), \quad j = 1,2.
\end{equation}
where $\gamma_2(\Lambda)$ is a constant depending on the lattice $\Lambda$, see \cite[eq. (8)]{H18}. Then we also notice that $n_i$ is the multiplicity of the non-repeated zeros $w_i$. As a result, we can rewrite $g(z)$ as
\[
    g(z) = ce^{p(z)} \prod_{i = 1}^{p_0}\left(\sigma(z - w_i)\right)^{n_i} = \tilde c e^{\tilde p(z)} \prod_{i = 1}^{p_0}\left(\tilde \sigma(z - w_i)\right)^{n_i} = \tilde c e^{\tilde p(z)} \prod_{i = 1}^{p}\tilde \sigma(z - z_i)
\]
for some other second degree polynomial $\tilde p(z)$. 
Denote $\tilde p(z) = \gamma_2 z^2 + \gamma_1 z$ where the constant term is absorbed by $\tilde c$. Using the boundary constraints of $g$ in \eqref{eq: bdry_cond} and of $\tilde \sigma$ in \eqref{eq: sigma_bdry}, we find for $j = 1,2$, 
    \begin{align}
    & 2\gamma_2 v_j +\xi_j = \xi_j,  \qquad \Rightarrow \qquad \gamma_2 = 0, \\
    & \gamma_1v_j + i\pi p - \frac{\xi_j}{p}\sum\limits_{i = 1}^{p} z_i = \frac{i\bfk \overline{v_j}}{2}.    \label{eq: eqn_gamma_1_Z_p}
     \end{align}
Denote $Z_{p}:=\sum\limits_{i=1}^{p} z_i \in \CC$. From \eqref{eq: eqn_gamma_1_Z_p} with $j = 1,2$, we can solve for 
\begin{equation}
\label{eq:solutions}
 \begin{split}
    \gamma_1 &= \frac{ip \pi(\xi_2 - \xi_1)-\frac{\bfk i }{2}(\xi_2\overline{v_1} - \xi_1\overline{v_2})}{\xi_1 v_2 - \xi_2 v_1} = i\frac{\bfk }{2} +i\frac{\pi  B}{3}(\lambda_2 - \lambda_1),\\
    Z_{p} &= \frac{ ip\pi (v_1 - v_2)-\frac{\bfk i }{2}(\overline{v_2}v_1 - \overline{v_1}v_2)}{\frac{1}{p}(\xi_2v_1 - \xi_1v_2)} =  \frac{2\pi i p }{3}(\lambda_2\omega^2 - \lambda_1\omega)-i\frac{\bfk p }{B} .
\end{split}
\end{equation}
In summary, $u_\bfk(z) = e^{-\frac{i\bfk \bar{z}}{2}} e^{B\frac{(z^2 - |z|^2)}{4}}g(z)\in \ker_{\mathcal H_f} (a_\bfk)$ as long as $g(z) = \tilde ce^{\gamma_1 z} \prod\limits_{i = 1}^{p} \tilde \sigma(z - z_i)$ with $\gamma_1$, $Z_p=\sum\limits_{i=1}^{p} z_i$ satisfy \eqref{eq:solutions}. 

Now we determine the dimension of this space. Since the function exhibits precisely $p$ zeros, we can construct also $p$ linearly independent solutions, see e.g. \cite[Lemma $4.1$]{BHZ} for an analogous construction. There cannot be more solutions by the Riemann-Roch theorem:

Indeed, consider one possible solution $u_{\bfk}$ with fixed zeros at $z_1,...,z_p$ and any other solution $v_{\bfk},$ then $\frac{v_{\bfk}}{u_{\bfk}}$ is a meromorphic function with (possible) poles at $z_1,...,z_p.$ The Riemann-Roch theorem states that for the special case of the two-dimensional torus the dimension of the space of meromorphic functions with this fixed set of possible poles is equal to the number of poles, which is $p$.

For completeness, we recall that the non-zero energy bands of the magnetic Dirac operator are then derived by defining first
\[ u_{\mathbf k,n} := \frac{(a^*_{\mathbf k})^n}{(2B)^{n/2} \sqrt{n!}} u_{\mathbf k}, \quad n \in \mathbb N_0,\] with $u_{\mathbf k} \in \ker(a_{\mathbf k}).$
We then find 
\[\begin{split}
    a_{\mathbf k}^*u_{\mathbf k,n} &= \sqrt{2B(n+1)} \ u_{\mathbf k,n+1} \text{ and }
     a_{\mathbf k} u_{\mathbf k,n} = \sqrt{2Bn} \ u_{\mathbf k,n-1}
    \end{split}.\]
This way, for $n \neq 0,$ 
\begin{equation}
\label{eq:indices}
\begin{split}
H^B_{D,\mathbf k} \begin{pmatrix} \pm u_{\mathbf k,n-1} \\ u_{\mathbf k,n} \end{pmatrix} &= \pm \sqrt{2Bn} \begin{pmatrix} \pm u_{\mathbf k,n-1} \\ u_{\mathbf k,n} \end{pmatrix}.
\end{split}
\end{equation}

We can then use this to see that $a_{\mathbf k}$ is Fredholm on $\mathcal H_f$, i.e. possesses a finite-dimensional kernel and has closed range. This follows for example from the criterion stated in \cite[p.158 (footnote)]{T92}:
The Hamiltonian $H_{D,\mathbf k}^B$ is Fredholm if and only if both $a_{\mathbf k}$ and $a_{\mathbf k}^*$ are. Meanwhile, by \cite[Theo $5.19$]{T92}, $H_{D,\mathbf k}^B$ is Fredholm if and only if $e^{-(H_{D,\mathbf k}^B)^2t}$ is trace-class for some $t>0$. The latter can be easily verified using \eqref{eq:indices}.
The Fredholm index 
\[ \operatorname{ind}_{\mathcal H_f}(a_{\mathbf k}) = \operatorname{dim}\ker_{\mathcal H_f}(a_{\mathbf k}) - \operatorname{dim}\ker_{\mathcal H_f}(a_{\mathbf k}^*)\]
is readily computed from the dimension of the nullspaces of $a_{\mathbf k}$, which is $p$ and of $a_{\mathbf k}^*$, which is $0.$
\end{proof}
\begin{rem}\label{remark: negative_mag}
Let $(Qu)(z) := \overline{u(-z)}$. Recall that $A(z) = \frac{Bi}{2}z$. Then by direct computation,
\[
Q (a_\bfk) Q^* = Q(2D_{\overline{z}} - A(z) + \bfk)Q =  2D_z + \overline{A(z)} + \overline{\bfk}.
\]
Thus $Qu_\bfk\in \ker_{\mathcal H_f}(2D_z + \overline{A(z)} + \overline{\bfk})$ for any $u_\bfk \in \ker_{\mathcal H_f}(a_\bfk)$. Notice that $2D_z + \overline{A(z)} + \overline{\bfk}\neq a_\bfk^* = 2D_z - \overline{A(z)} + \overline{\bfk}$. In fact, $2D_z + \overline{A(z)} + \overline{\bfk} = (a_\bfk(-B))^*$ where $a_\bfk(-B) = 2D_{\bar{z}} + \frac{Bi}{2}z$ represents the negative constant magnetic field $-B$. 
\end{rem}

\begin{rem}\label{rem: protected_states}
If $\lambda_1=\lambda_2=1$ the last two equations \eqref{eq:solutions} reduce to 
\begin{equation}
\label{eq:magnetic_param}
 \gamma_1 = \frac{i\mathbf k}{2} \text{ and }Z_p = \frac{2\pi p}{\sqrt{3}} - \frac{4\pi i }{3\sqrt{3}} \mathbf k.
 \end{equation}
\end{rem}

A direct consequence of Theorem \ref{theo:Fredholm} is, we can compute the Fredholm index for the chiral model. 
\begin{corr}\label{cor: index_D_c}
Under the same assumptions as in Theorem \ref{theo:Fredholm},
    the operator $D_c(\alpha_1, B)+ \mathbf k$ on $\mathcal H_f$ is a Fredholm operator of index $2p.$
\end{corr}
\begin{proof}
We recall that the direct sum of Fredholm operators  $\diag(a_\bfk, a_\bfk)$ is a Fredholm operator with twice the index of $a_{\mathbf k}$. Since $U,U_-$ are smooth functions, $\begin{pmatrix} 0 & \alpha_1 U(z) \\ \alpha_1 U_-(z) & 0 \end{pmatrix}$ is a relatively compact perturbation \cite[Prop.$5.27$]{T92} of $\diag(a_\bfk, a_\bfk)$ which does not affect the index. To summarize, we have argued that 
\begin{equation}
\label{eq:index_formula}
\operatorname{ind}_{\mathcal H_f}(D_c(\alpha_1, B)+ \mathbf k) =2 \operatorname{ind}_{\mathcal H_f}(a_{\bfk})= 2p.
\end{equation}
\end{proof}

We also observe that every zero of elements in the nullspace of $D_c(\alpha_1, B)+k$ are of $z$-type (and analogously for elements in the nullspace of $D_c(\alpha_1, B)^*+\bar k$ are of $\bar z$-type). 
\begin{lemm}
\label{lemm:type_of_zeros}
    Let $w \in C^{\infty}(\CC;\CC^2)$ and that $(D_c(\alpha_1, B)+\bfk)w=0$ for some $\bfk \in \CC$ and that $w(z_0)=0.$ Then $w(z)=(z-z_0)w_0(z)$ with $w_0 \in C^{\infty}(\CC;\CC^2).$ 
\end{lemm}
\begin{proof}
   It suffices to show that $(2D_{\bar z})^{\ell}w(z_0)=0$ for all $\ell$ which follows from $(2D_{\bar z})^{\ell}w(z)=(2D_{\bar z})^{\ell-1}\begin{pmatrix}
       A - \bfk & -U \\ -U_- & A - \bfk
   \end{pmatrix} w(z)$ by induction.
\end{proof}

\section{The chiral limit}
\label{sec:chiral_lim}
We now turn to the chiral model \eqref{eq:DcB} and study its flat bands at energy zero. An important quantity characterizing the integer Quantum Hall effect is the Chern number.
To this end, we also introduce the Chern number of a spectral projection.

Let $\mathbb C/\Gamma_{\operatorname{mag}}^* \ni \mathbf k \mapsto P(\mathbf k):L^2(\CC/\Gamma_{\text{mag}}) \to L^2(\CC/\Gamma_{\text{mag}})$ be a smooth family of projections, then the associated Chern number is given by
\[ c_1(P) := \int_{\mathbb C/\Gamma_{\operatorname{mag}}^*} \tr_{L^2(\mathbb C/\Gamma_{\operatorname{mag}})}(P(\mathbf k)[\partial_1 P(\mathbf k),\partial_2 P(\mathbf k)]) \frac{d\mathbf k}{2\pi i}.\]

\subsection{Flat bands away from magic angles}
Recall that $\mathcal {A}$ denotes the set of magic angles without a magnetic field (see Subsection \ref{subsec: charofmagicangles}). Here in this subsection, we give a full characterization of flat bands with magnetic field when $\alpha_1 \notin \mathcal A$. In the next subsection, we will discuss the case when $\alpha_1\in \mathcal A$.

\begin{figure}
    \centering
    \includegraphics[width=6cm]{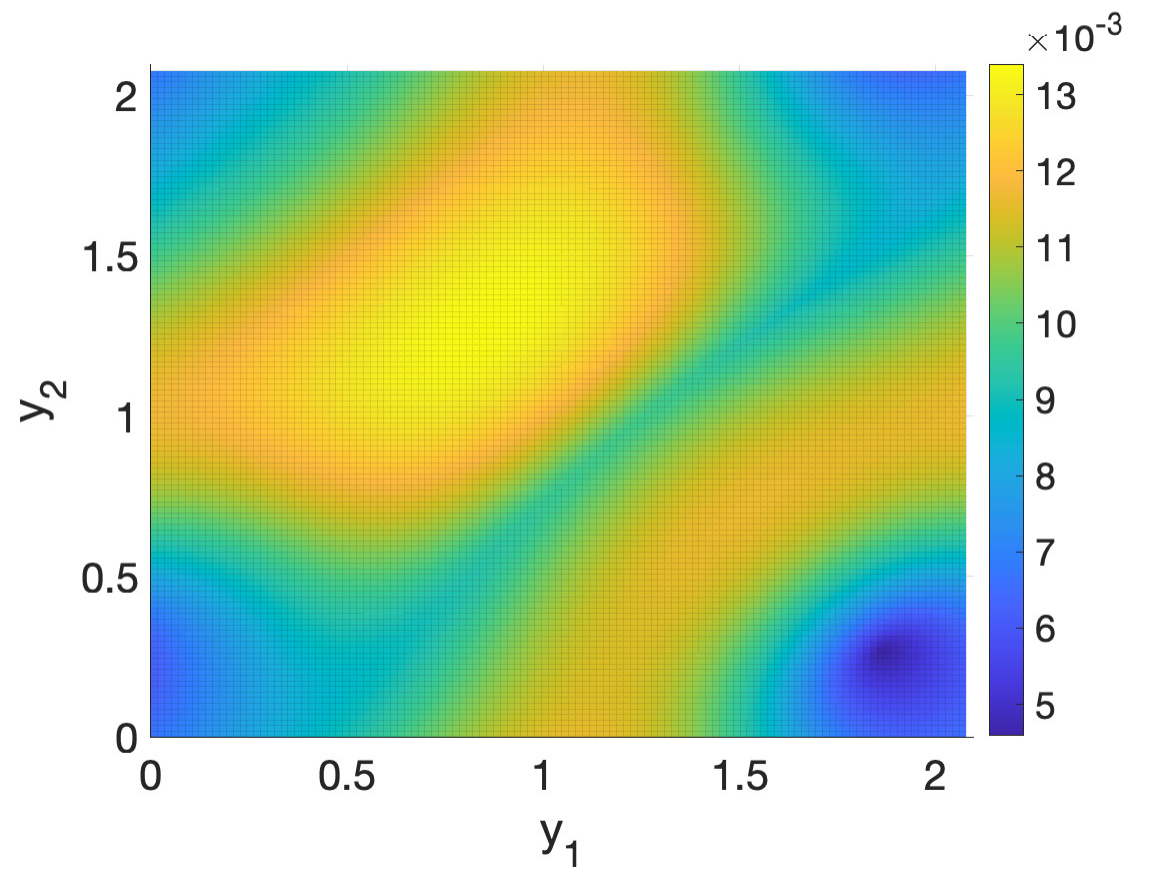}
    \includegraphics[width=6cm]{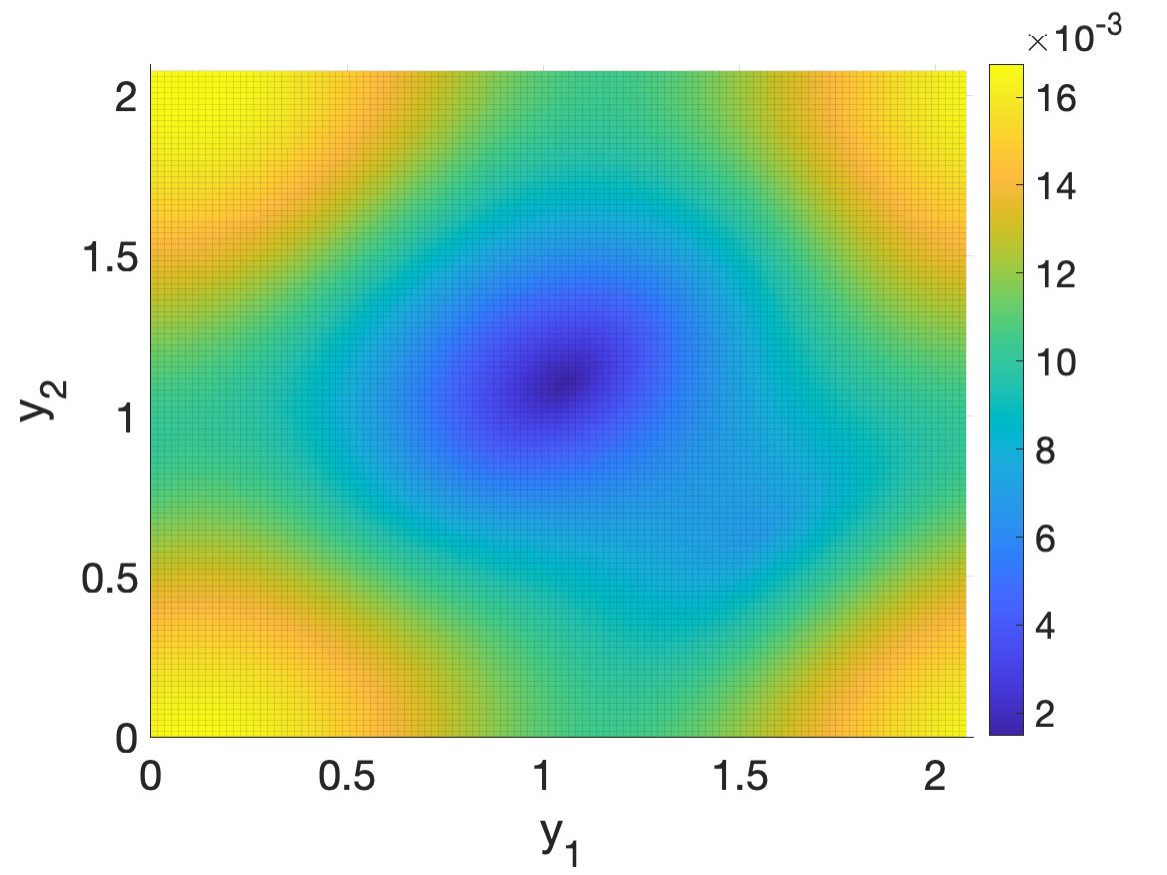}
    \caption{Figure showing modulus of two flat bands at $\alpha=0.2$ (\emph{non-magic}) away from magic angles with magnetic flux $\Phi=2\pi.$ }
    \label{fig:my_label3}
\end{figure}

\begin{theo}[Flat bands; $\alpha_1 \notin \mathcal A$]
\label{theo:away_mag_angl}
    Assume $A(z)$ satisfy Assumption \ref{assump: magnetic potential} with $A_{\operatorname{per}}=0$ and $B>0$, let $\alpha_1 \notin \mathcal A,$ then $H_c(\alpha_1, B)$ has a flat band of multiplicity $2p$ at energy $0$ that is uniformly (in $\bfk$) gapped away from the rest of the spectrum. In particular, for all $\mathbf k \in \CC$
    \begin{align}
\ker_{\mathcal H_f}(H_{c,\mathbf k}(\alpha_1, B)) &=\ker_{\mathcal H_f}(D_c(\alpha_1,B)+\mathbf k) \times \{0_{\CC^2}\} \label{eq: kerH_c,k}\\
     \ker_{\mathcal H_f}(D_c(\alpha_1, B)+\mathbf k) &=(u_1(\alpha_1) \otimes \ker_{\mathcal H_f}(a_{\mathbf k-\mathbf k_0})) \oplus (u_2(\alpha_1) \otimes \ker_{\mathcal H_f}(a_{\mathbf k})),     \label{eq: kerD_c}
     \end{align}
    where $\alpha_1 \mapsto u_1(\alpha_1),u_2(\alpha_1)$ are two real-analytic $\mathcal H_f(B=0,\Gamma)$ functions with $\mathbf k_0:=\frac{1}{3}(\eta_1 +\eta_2) = -i$, where $\eta_1,\eta_2$ is the dual basis as in Subsection \ref{subsec: BM_model}. Here, $ u_1(\alpha_1), u_2(\alpha_1) \in \ker_{L^2_0}(D_c(\alpha_1,B=0))$ are the unique real-analytic functions satisfying \begin{equation}
\label{eq:identity}
\begin{pmatrix} 2D_{\bar z} + \mathbf k_0 & 0 \\ 0 &  2D_{\bar z} + \mathbf k_0 \end{pmatrix}  u_1(\alpha_1=0) =0, \quad  \quad \begin{pmatrix} 2D_{\bar z} & 0 \\ 0 & 2D_{\bar z} \end{pmatrix} u_2(\alpha_1=0) =0.
 \end{equation} 
 In addition, for $\alpha_1 \notin \mathcal A$, the Chern number of the $2p$ flat bands is $-2$.
\end{theo}
\begin{proof}
To show \eqref{eq: kerH_c,k}, by \eqref{eq:DcB}, it is enough to show 
\begin{equation}
    \label{eq: kerD_c*}
    \ker_{\mathcal H_f}(D_c(\alpha_1, B)^* + \overline{\bfk}) = \{0\},\quad \forall \bfk \in \CC. 
\end{equation}
Assume for some $\bfk_1 \in \CC$, there is $\varphi_{\overline{\bfk_1}}\neq 0 \in \ker_{\mathcal H_f}(D_c(\alpha_1, B)^* + \overline{\bfk_1})$.
By Remark \ref{remark: negative_mag}, for any $\bfk$, there is $Qu_{\bfk - \bfk_1}\in \ker_{\mathcal H_f}(2D_z + \overline{A(z)} + \overline{\bfk} - \overline{\bfk_1})$. Then $\psi_\bfk:=\varphi_{\overline{\bfk_1}} (Qu_{\bfk - \bfk_1}) $ satisfies 
\begin{equation}
    \label{eq: product}
\begin{split}
    D_c(\alpha_1, B = 0)^*\psi_\bfk
 & = \left[(D_c(\alpha_1, B)^* + \overline{\bfk_1} + \overline{A(z)} - \overline{\bfk_1})\varphi_{\overline{\bfk_1}} \right]Qu_{\bfk - \bfk_1} + \varphi_{\overline{\bfk_1}}(2D_z Qu_{\bfk - \bfk_1})\\
 & = (\overline{A(z)} - \overline{\bfk_1})\psi_\bfk + (\overline{\bfk_1} - \overline{\bfk} - \overline{A(z)})\psi_\bfk = -\overline{\bfk}\psi_{\bfk}.
\end{split}
\end{equation}
This implies that $\psi_{\bfk}\in \ker_{\mathcal H_f}(D_c(\alpha_1, B = 0)^* + \overline{\bfk})$ for all $\bfk$; thus, $\alpha_1\in \mathcal A$. We get a contradiction.

To show \eqref{eq: kerD_c} and \eqref{eq:identity}. We first notice that for $\bfk_0 = \frac{1}{3}(\eta_1 + \eta_2) =\frac{\omega^2-\omega}{\sqrt{3}}=\frac{2i}{\sqrt{3}} \Im(\omega^2)= -i$,
\[
u_1(\alpha_1 = 0) = e^{-i\langle z, \bfk_0\rangle}\begin{bmatrix}
    1 \\ 0
\end{bmatrix}, \qquad u_2(\alpha_1 = 0) = \begin{bmatrix}
    0\\ 1
\end{bmatrix}
\]
satisfy \eqref{eq:identity} directly. Furthermore, we can verify $u_1(\alpha_1 = 0) \otimes \ker_{\mathcal H_f}(a_{\bfk - \bfk_0}) \in \mathcal H_f(\alpha_1 = B = 0, \Gamma)$ by noticing that the Abelian group that commutes with $D_c(\alpha_1 = B = 0, \Gamma)$ is
\[
\mathcal L_\bfn = \diag(\omega^{n_1 + n_2},1) \otimes T_1^{n_1}T_2^{n_2}, \quad T_j u(z) = u(z + \zeta_j) 
\]
and $u_1(\alpha_1 = 0) \otimes u_{\bfk - \bfk_0}$ is invariant under $\mathcal L_\bfn$ because
\[
\omega T_j(u_1(\alpha_1 = 0)) = \omega e^{-i\langle \zeta_j, \bfk_0\rangle} u_1(\alpha_1 = 0) =  u_1(\alpha_1 = 0), \quad j = 1,2, 
\]
where we used $\omega = \exp(\frac{2\pi i}{3})$ and  $\langle \zeta_j, \bfk_0\rangle = \frac{1}{3}\langle \zeta_j, \eta_1 + \eta_2\rangle = \frac{2\pi}{3}$. Similar arguments work for $u_2(\alpha_1 = 0)\otimes \ker_{\mathcal H_f}(a_\bfk)$. 

By the symmetry arguments and Rellich's theorem in \cite[Prop. $2.2$]{BEWZ20a}, there are real analytic families $\mathbb R \ni \alpha_1 \mapsto u_1(\alpha_1)$ and $\mathbb R\ni \alpha_1 \mapsto u_2(\alpha_1)$ such that $u_1(\alpha_1)\in \ker_{\mathcal H_f}(D_c(\alpha_1, B = 0) + \bfk_0), u_2(\alpha_1)\in \ker_{\mathcal H_f}(D_c(\alpha_1, B = 0))$. As a result, we have
\begin{equation}
    \label{eq: space_included}
(u_1(\alpha_1) \otimes \ker_{\mathcal H_f}(a_{\mathbf k-\mathbf k_0})) \oplus (u_2(\alpha_1) \otimes \ker_{\mathcal H_f}(a_{\mathbf k}))\subset \ker_{\mathcal H_f}(D_c(\alpha_1,B)+\mathbf k).
\end{equation} 
In fact, if $u_{\bfk - \bfk_0}\in  \ker_{\mathcal H_f}(a_{\mathbf k-\mathbf k_0})$, then $\phi_\bfk = u_1(\alpha_1)u_{\bfk - \bfk_0}$ satisfies
\[
\begin{split}
    D_c(\alpha_1, B)(\phi_\bfk) &= \left[(D_c(\alpha_1, B = 0) + \bfk_0 - A(z)- \bfk_0)u_1(\alpha_1)\right]u_{\bfk - \bfk_0} + u_1(\alpha_1)2D_{\overline{z}}u_{\bfk - \bfk_0}\\
    &= (-A(z) - \bfk_0)\phi_\bfk + (\bfk_0 - \bfk + A(z))\phi_\bfk = -\bfk \phi_\bfk.
\end{split}
\]
Similar argument works for $u_2(\alpha_1)\otimes \ker_{\mathcal H_f}(a_\bfk)$. Hence we get \eqref{eq: space_included}. 

To get equality \eqref{eq: kerD_c}, we notice that, by Theorem \ref{theo:Fredholm}, $\dim(a_\bfk) = p$; thus
\[
\dim (u_1(\alpha_1) \otimes \ker_{\mathcal H_f}(a_{\mathbf k-\mathbf k_0})) \oplus (u_2(\alpha_1) \otimes \ker_{\mathcal H_f}(a_{\mathbf k})) = 2p. 
\]
While by Cor. \ref{cor: index_D_c} and \eqref{eq: kerD_c*}, 
\[
\dim(\ker_{\mathcal H_f}(D_c(\alpha_1,B)+\mathbf k)) = \operatorname{ind}(\ker_{\mathcal H_f}(D_c(\alpha_1,B)+\mathbf k)) - \dim(\ker_{\mathcal H_f}(D_c(\alpha_1,B)^*+\overline{\bfk}) = 2p.
\]
Thus the inequality in \eqref{eq: space_included} is actually an equality. Thus we get \eqref{eq: kerD_c}.

In particular, since bands depend continuously on $\mathbf k$ and $\dim \ker(D_c(\alpha_1,B)+\mathbf k) = 2p$ is independent of $\mathbf k$, the $2p$ flat bands of the Hamiltonian at zero energy are gapped away from the rest of the spectrum.

By norm continuity of the spectral projection $\mathcal A\ni \alpha_1 \mapsto \indic_{\{0\}}(H_c(\alpha_1, B))$, the Chern number does not depend on $\alpha_1,$ see \cite{B87}. Using the value of the Chern number for $\alpha_1=0$, see \cite[Lemma $5$]{BES94} and the additivity of Chern numbers under direct sums of projections, see \cite{B87}, implies the claim.
\end{proof}
For general constant magnetic fields, we conclude
\begin{corr}
  Let $B>0$ be a constant magnetic field. Let $\alpha_1 \notin \mathcal A,$ then the Hamiltonian exhibits an eigenvalue of infinite multiplicity at energy $0$, that is gapped away from the rest of the spectrum, such that 
    \begin{equation}
    \begin{split}
    \label{eq:nullspace4}
\ker_{L^2(\CC)}(H_{c}(\alpha_1, B)) &=\ker_{L^2(\CC)}(D_c(\alpha_1,B)) \times \{0_{\CC^2}\}\\
\ker_{L^2(\CC)}(D_c(\alpha_1,B))&= (u_1(\alpha_1) \otimes \ker_{L^2(\CC)}(a)) \oplus (u_2(\alpha_1) \otimes \ker_{L^2(\CC)}(a))
     \end{split}
     \end{equation}
     with $\ker_{L^2(\CC)}(a)$ as in \eqref{eq:nullsp} and $u_1,u_2$ as in Theorem \ref{theo:away_mag_angl}.
 In addition, for $\alpha_1 \notin \mathcal A$, the Chern number of the eigenvalue at zero energy is $-2$. 
\end{corr}
\begin{proof}
By the previous Theorem \ref{theo:away_mag_angl}, it follows that for every rational magnetic flux, the flat band is gapped from the rest of the spectrum. By extending the gap stability of magnetic Schr\"odinger operators to the chiral Hamiltonian, see Avron-Simon \cite{AvSi}, Nenciu \cite{N86}, Sj\"ostrand \cite[Prop.$2.4$]{Sj89}, it follows that the gap persists to any magnetic flux. One can then use the Streda formula \cite{B87} and \cite[Prop. $3.2$]{BKZ22} to conclude that the Chern number of the spectrum at zero energy remains $-2.$
\end{proof}    
The previous Theorem \ref{theo:away_mag_angl}, can be easily extended when there are additional periodic magnetic potentials. 

Indeed, any $L^2(\CC/\Gamma_{\operatorname{mag}})$ function $f$ is, for some suitable coefficients $(f_{\mathbf k}) \in {\ell^2(\Gamma^*_{\operatorname{mag}})}$, of the form $f(z)=\sum_{\mathbf k\in\Gamma_{\operatorname{mag}}^*}f_{\mathbf k} e^{i\Re(\mathbf k \bar z)}.$ Such a magnetic potential has a Fourier expansion 
\begin{equation}
\label{eq:A_per}
A_{\operatorname{per}}(z)=\sum_{\mathbf k \in\Gamma^*_{\operatorname{mag}} \setminus\{0\}}A_{\mathbf k} e^{i\Re(\mathbf k \bar z)}\in L^2(\mathbb C/\Gamma_{\operatorname{mag}}).
\end{equation}
\begin{prop}
\label{prop:periodic magnetic field}
Consider a lattice $\Gamma_{\operatorname{mag}}$ and the setting of Theorem \ref{theo:away_mag_angl} with additional periodic vector potential $A_{\operatorname{per}} \in C^{\infty}(\CC/\Gamma_{\operatorname{mag}})$ with $A_{\operatorname{per}}(0)=0$. Let $\alpha_1 \notin \mathcal A,$ then the Hamiltonian exhibits a flat band of multiplicity $2p$ at energy zero, that is uniformly in $\bf k$ gapped away from the rest of the spectrum. Thus, for all $\mathbf k \in \CC$ with $u_1,u_2$ and $\mathbf k_0$ as in \eqref{eq: kerD_c}
 \begin{equation}
     \begin{split}
\ker_{\mathcal H_f}(H_{c,\mathbf k}) &=\ker_{\mathcal H_f}(D_c(\alpha_1,B)+\mathbf k) \oplus \{0_{\CC^2}\}\\
  \ker_{\mathcal H_f}(D_c(\alpha_1, B)+\mathbf k) &=(u_1(\alpha_1) \otimes \ker_{\mathcal H_f}(a_{\mathbf k-\mathbf k_0})) \oplus (u_2(\alpha_1) \otimes \ker_{\mathcal H_f}(a_{\mathbf k})), 
       \end{split}
 \end{equation}
    where $\alpha_1 \mapsto u_1(\alpha_1),u_2(\alpha_1)$ are two real-analytic $L^2(\CC/\Gamma_{\operatorname{mag}})$ functions.
 In addition, for $\alpha_1 \notin \mathcal A$, the Chern number of the $2p$ flat bands is $-2.$
\end{prop}
\begin{proof}
  Notice that for arbitrary $\phi$, we have $e^{-\phi}a_\bfk e^{\phi} = a_\bfk - 2i\partial_{\bar{z}}\phi(z).$
   If we can find $\phi$ s.t. $
    2i\partial_{\bar{z}} \phi(z) = A_{\operatorname{per}},$
    then for $\psi_{\bfk}$ solving $a_{\mathbf k}\psi_{\bfk} = 0$ we find $ (a_{\bfk} - A_{\operatorname{per}}) e^{-\phi}\psi_{\bfk}=0.$
   However, by using the Fourier series, it is easy to see that 
$
     \phi(z) =-i \sum_{\mathbf k \in\Gamma^*_{\operatorname{mag}} \setminus\{0\}}A_{\mathbf k} \frac{e^{i\Re(\mathbf k \bar z)}}{\mathbf k}$
   is a solution. Thus, zero mode solutions for non-vanishing periodic magnetic potentials are in one-to-one correspondence with solutions for zero periodic magnetic potential. Finally, by switching on the periodic potential gradually, norm continuity of the Chern number implies that the Chern numbers also coincide.  
\end{proof}

\subsection{Flat bands at magic angles}
In Theorem \ref{theo:away_mag_angl}, we showed that the chiral Hamiltonian with commensurable, i.e. rational, constant magnetic field exhibits $2p$ flat bands at energy zero that are gapped away from the rest of the spectrum. We shall now study what happens at magic angles (of the non-magnetic Hamiltonian). To do so, we need some preliminary discussion of the flat band wave function in the absence of magnetic fields 

\subsubsection{Preliminaries}
\begin{defi}[Multiplicity of magic angles]
We call a magic angle $\alpha_1 \in \mathcal A$ of the non-magnetic Hamiltonian $n$-fold degenerate, if $\dim \ker_{L^2(\CC/\Gamma)}(D_c(\alpha_1,B=0)+\mathbf k)=n$ for all $\bfk \in \CC.$ In the case of $n=1$, we call the magic angle also \emph{simple}.
\end{defi}

Our main focus is on simple and two-fold degenerate magic angles, since they are the only ones that generically occur in the non-magnetic chiral limit of twisted bilayer graphene (see \cite[Theo.$3$]{BHZ3}). 

Notice that a simple magic angle $\alpha_1\in \mathcal A$ in our definition is equivalent to the flat band at $0$ being two-fold degenerate for the chiral model $H_c(\alpha_1, B = 0)$. This is due to the conjugate relation 
\[
    Q(D_c(\alpha_1, B) + \bfk)Q = D_c(\alpha_1, -B)^* + \overline{\bfk}
\]
where $Qu(z) = \overline{u(-z)}$ and the Definition \eqref{eq:DcB}.

If a magic angle $\alpha_1\in \mathcal A$ is simple, then $u_{\mathbf 0} := u_\bfk|_{\bfk = \mathbf 0}\in \ker_{\mathcal H_f(B=0)} D_c(\alpha_1, B=0)$ has a unique zero $u_{\mathbf 0}(-z_S)=0$ where $z_S:=\frac{4}{9}\pi i (\omega^2-\omega)=\frac{4\pi}{3\sqrt{3}}$, see \cite[Theo.$2$]{BHZ2}; if a magic angle $\alpha_1\in \mathcal A$ is two-fold degenerate, then any $u_{\mathbf 0}\in \ker_{\mathcal H_f(B=0)} D_c(\alpha_1, B=0)$ has two zeros $z_1,z_2$ with $z_1+z_2=-z_S$. 

We now move on how to construct elements in $\ker_{\mathcal H_f}(D_c(\alpha_1, B = 0) + \bfk)$, see \cite{BHZ2} for more details. We recall the definition of the Jacobi $\theta$-function:
\[\begin{split} \theta_1(\zeta\vert\omega) :=-\theta_{1/2,1/2}(\zeta\vert \omega) = -\sum_{n \in \ZZ}e^{\pi i(n+1/2)^2 \omega + 2\pi i(n+1/2)(\zeta+1/2)},\\
\theta_1(\zeta+m\vert \omega) = (-1)^m\theta_1(\zeta\vert m), \quad \theta_1(\zeta+n\omega \vert \omega) = (-1)^n e^{-\pi i n^2 \omega - 2\pi i \zeta n} \theta_1(\zeta\vert \omega).\end{split}\]
The function $\theta_1$ has zeros at the lattice points $\ZZ + \omega \ZZ$. We then define the meromorphic function 
\[ g_k(\zeta) = e^{2\pi  \zeta (k-\bar k)/\sqrt{3}} \frac{\theta_1(\zeta+k\vert\omega)}{\theta_1(\zeta\vert\omega)}\text{ with } k=\omega a-b, a,b \in \RR \]
and introduce $g_{\mathbf k}(z):=g_k(3z/(4\pi i \omega))$, with $\bfk = \sqrt{3}\omega k$, satisfying 
\[ g_{\mathbf k}(z+ a) = e^{i\langle  a, \mathbf k \rangle} g_{\mathbf k}(z), \quad a \in \Gamma.\]
We thus define $u_{\mathbf p}(z) := g_{\mathbf p}(z+z_S)u_{\mathbf 0}(z),$ for  $$\mathbf p \in  \Gamma^*_{\operatorname{mag}}/\Gamma^* = \Big\{ \tfrac{p_1\sqrt{3}\omega^2}{\lambda_1} -\tfrac{p_2 \sqrt{3}\omega}{\lambda_2},p_1 \in \ZZ_{\lambda_1},p_2\in \ZZ_{\lambda_2}\Big\}.$$ This implies that $$\ker_{\mathcal H_f(B=0,\Gamma_{\operatorname{mag}})}(D_c(\alpha_1,B=0)) = \operatorname{span}\{u_{\mathbf p}; \mathbf p \in \Gamma^*_{\operatorname{mag}}/\Gamma^*\}$$ with functions $u_{\mathbf p},$ as above, satisfying
\[D_c(\alpha_1,B=0) u_{\mathbf p} =0, \quad \mathscr L_a u_{\mathbf p} = e^{2\pi i\big(\tfrac{p_1a_1}{\lambda_1}+\tfrac{p_2a_2}{\lambda_2}\big) } u_{\mathbf p} \text{ with }p_1 \in \ZZ_{\lambda_1},p_2 \in \ZZ_{\lambda_2}.\]
Thus, we see that each function in $\ker_{\mathcal H_f(B=0,\Gamma_{\operatorname{mag}})}(D_c(\alpha_1,B=0)+\bf k)$ has precisely 
\begin{equation}
    \label{eq: number_of_zeros}
\vert \Gamma/ \Gamma_{\operatorname{mag}}\vert = q
\end{equation}
many zeros. The same is true for the nullspace of $D_c(\alpha_1,B=0)^*+\overline{\bfk}$ due to the symmetry 
\begin{equation}
\label{eq:symmetry}
Q(D_c(\alpha_1,B)+\bfk ) Q = D_c(\alpha_1,-B)^*+\overline{\bfk}
\end{equation}
with anti-linear symmetry $Qv(z) = \overline{v(-z)}.$ 

In a similar way, using
$F_{\bfk}(z):=-\frac{1}{\pi}\Big(\tfrac{2}{3}\Big)^3 F_k\Big(\frac{3z}{4\pi i \omega}\Big),$ for
$F_k(\zeta) = e^{\frac{2\pi (\zeta-\bar \zeta)k}{\sqrt{3}}} \frac{\theta_1(\zeta+k\vert \omega)}{\theta_1(\zeta\vert \omega)},$ we get the solution 
\[(D_c(\alpha_1,B=0)+\mathbf k)v_{\mathbf k}=0\]
for magic $\alpha_1 \in \mathcal A$ by setting 
\begin{equation}
\label{eq:translation_prop}
    v_{\bfk}(z)=F_{\bfk}(z+z_S)u_{\textbf{0}}(z).
\end{equation}
We recall that by Rellich's theorem \cite[Theo $3.9$]{Ka}, eigenvectors of $H_{c,\bfk}(\alpha_1)$ can be chosen as linearly independent real-analytic functions of $\alpha_1$. Thus, we can find an orthonormal set $X_{\mathbf k}$ of $2p$ normalized Bloch functions 
\begin{equation}
\label{eq:Bloch_frame}
 X_{\bfk}(\alpha_1):=\{u_{1,\bfk}(\alpha_1),....,u_{2p,\bfk}(\alpha_1)\}
 \end{equation}
such that $\dim \operatorname{span} X_{\bfk}=2p$ and each element of $X_{\bfk}(\alpha_1)$ depends real analytically on $\alpha_1 \in \mathbb R$ where $\operatorname{span} X_{\bfk} = \ker_{\mathcal H_f}(D_c(\alpha_1,B)+\mathbf k)$ for $\alpha_1 \notin \mathcal A.$

We can now analyze the zero band structure at non-magnetic magic angles. We call a function $f:\mathbb C \to \mathbb C^4$ $A$-sublattice polarized if it is of the form $f(z)=(f_1(z),f_2(z),0,0)$. A function $f(z)=(0,0,f_1(z),f_2(z))$ is called $B$-sublattice polarized. $A$-lattice polarized functions in the chiral limit also have the property of  \emph{vortexability}, see \cite{LVP23}, which is a criterion closely related to the fractional quantum Hall effect.

\begin{figure}
    \centering
    \includegraphics[width=6cm]{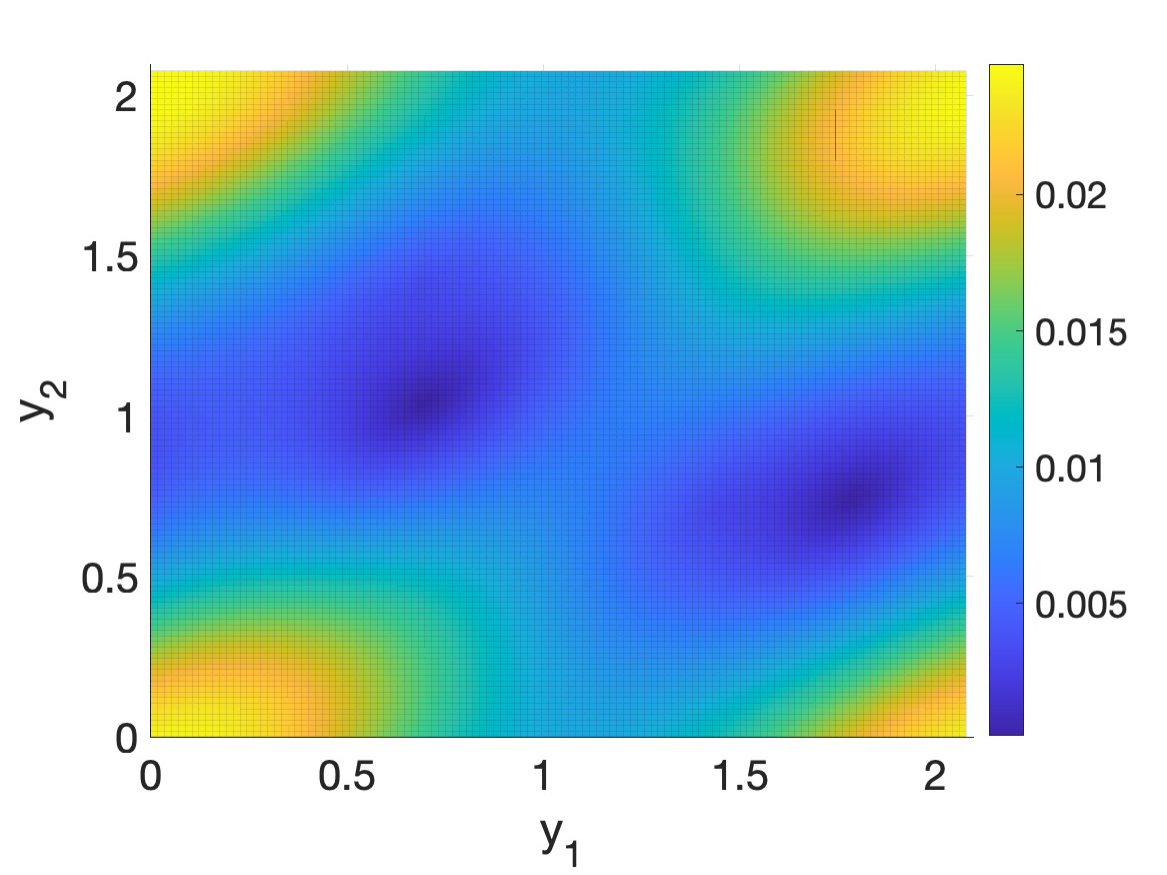}
    \includegraphics[width=6cm]{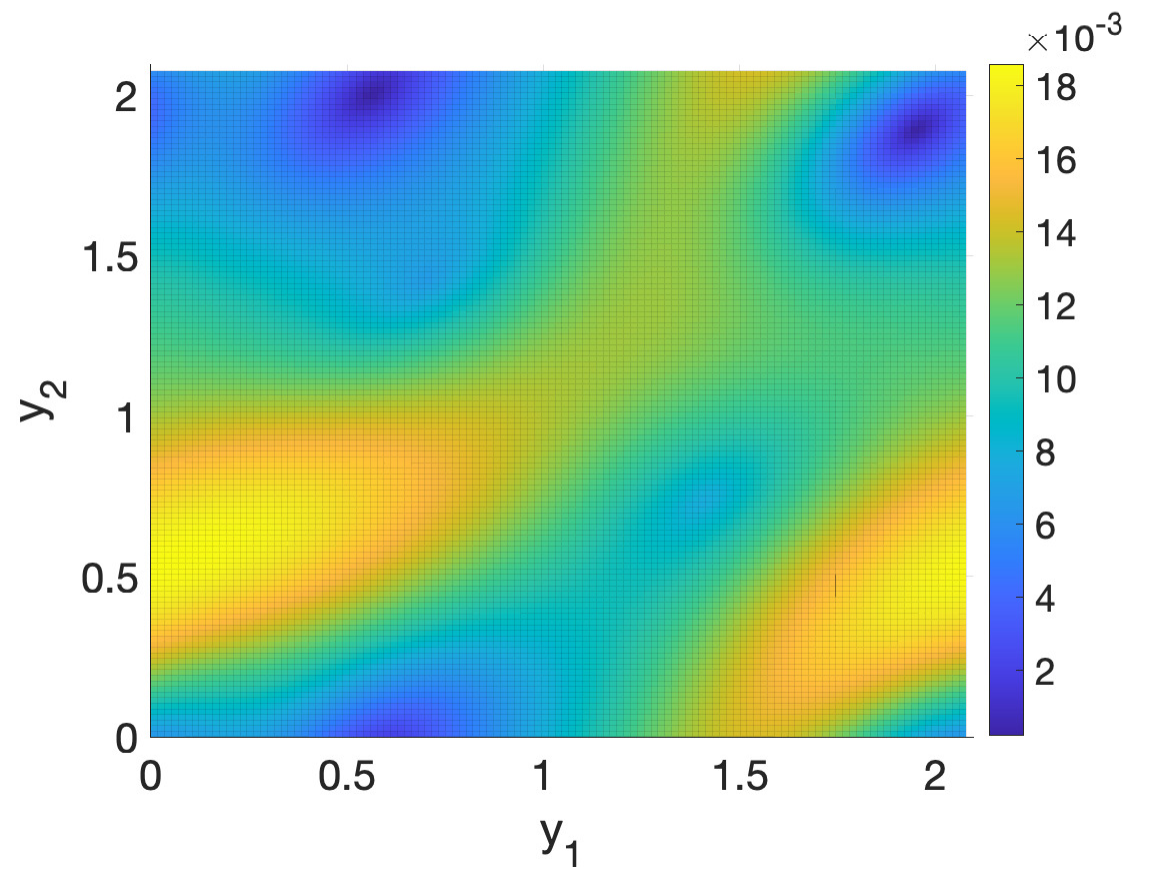}
    \caption{Figure showing modulus of two flat bands at $\alpha_1=0.5865$, the first magic angle, with $\Phi=2\pi.$ }
    \label{fig:my_label2}
\end{figure}
\begin{theo}[Simple magic angle]
   \label{theo:simple}
      Assume $A(z)$ satisfy Assumption \ref{assump: magnetic potential} with $A_{\operatorname{per}}=0$, $B>0$, and $\alpha_1 \in \mathcal A$ is a simple magic angle.
     \begin{enumerate}
     \item If $p>q$, all $2p$ flat bands at zero energy are $A$-lattice polarized 
     \begin{equation}
     \begin{split}
    \label{eq:nullspace2}
    &\ker_{\mathcal H_f}(H_{c,\mathbf k}) =\ker(D_c(\alpha_1,B)+\mathbf k) \times \{0_{\CC^2}\}\\
&\ker_{\mathcal H_f}(D_c(\alpha_1, B)+\mathbf k) =X_{\mathbf k}(\alpha_1) \text{ and }\ker_{\mathcal H_f}(D_c(\alpha_1, B)^*+\overline{\mathbf k})=\{0\}
     \end{split}
     \end{equation} 
     The Chern number of the $2p$ flat bands, that are, uniformly in $\bfk$, gapped away from the rest of the spectrum, is $-2$. 
     \item If $p=q$, and thus without loss of generality $p=q=1$, choose $\bfk_0 \notin \Gamma_{\text{mag}}^*$ and $u_1,u_2$ as in    \eqref{eq: kerD_c}. Then for $\eta_{\bfk-\bfk_0} \in \ker_{\mathcal H_f}(a_{\bfk-\bfk_0})\setminus\{0\}$ and $\eta_{\bfk} \in \ker_{\mathcal H_f}(a_{\bfk})\setminus\{0\}$
     \begin{equation}
    \label{eq:nullspace3}
     \ker_{\mathcal H_f}(D_c(\alpha_1, B)+\mathbf k) \supseteq \operatorname{span}\{u_1(\alpha_1) \eta_{\bfk-\bfk_0}, u_2(\alpha_1) \eta_{\bfk} \},
     \end{equation} 
    where equality holds for all points aside from $\bfk = i/2.$ 
     The Chern number of the two flat bands is $-2$. The zero energy flat bands are touched by bands from above and below (energetically) at $\bfk =i/2.$  All flat bands are $A$-lattice polarized.   
     \item Let $p<q$, then there are $2q$ bands at zero energy that are, uniformly in $\bfk$, gapped from the rest of the spectrum such that for all $\mathbf k \in \CC$
\[    \ker_{\mathcal H_f}(H_{c,\mathbf k})=\ker(D_c(\alpha_1,B)+\mathbf k) \times \{0_{\CC^2}\} 
    + \{0_{\CC^2}\} \times\ker(D_c(\alpha_1,B)^*+\overline{\mathbf k}) \]
    where $\operatorname{dim}\ker(D_c(\alpha_1,B)+\mathbf k) = p+q$ and $\operatorname{dim}\ker((D_c(\alpha_1,B)+\mathbf k)^*) = q-p.$
     The total Chern number of the flat bands is $0.$
     \end{enumerate}
\end{theo}
\begin{proof}
We shall use the same notation as in Theorem \ref{theo:away_mag_angl}. To prove (1), it suffices again, as in Theorem \ref{theo:away_mag_angl}, to show $\dim(\ker_{\mathcal H_f}(D_c(\alpha_1, B)^*+\overline{\mathbf k})) = 0$ for all $\bfk\in \CC$. This is because 
\begin{equation}
    \label{eq: inequality}
    X_\bfk(\alpha_1) \subset \ker(D_c(\alpha_1, B)+\bfk) \quad \text{~and~}\quad \dim(X_\bfk(\alpha_1)) = 2p
\end{equation}
by previous discussion \eqref{eq:Bloch_frame}. Once we can show $\dim(\ker_{\mathcal H_f}(D_c(\alpha_1,B)^*+\overline{\mathbf k})) = 0$, by the index formula \eqref{eq:index_formula}, we will have $\dim(\ker_{\mathcal H_f}(D_c(\alpha_1, B)+\bfk))=2p$. Then the inequality \eqref{eq: inequality} becomes equality. The Chern number and ``gapped away from other bands'' follow the same way as the proof of Theorem \ref{theo:away_mag_angl}.

To show $\dim(\ker_{\mathcal H_f}(D_c(\alpha_1,B)^*+\overline{\mathbf k})) = 0$, we assume for some $\bfk_1\in \CC$, there is $\varphi_{\overline{\bfk_1}}\neq 0 \in \ker_{\mathcal H_f}(D_c(\alpha_1, B)^* + \overline{\bfk_1})$. Then for all $\bfk$, by the same argument as Theorem \ref{theo:away_mag_angl}, (more explicitly, \eqref{eq: product}), we have $\psi_\bfk = \varphi_{\overline{\bfk_1}}(Qu_{\bfk - \bfk_1})\in \ker(D_c(\alpha_1,B=0)^*+\overline{\mathbf k})$. However, $Qu_{\bfk - \bfk_1}$ has $p$ many $\bar z$-type zeros by Lemma \ref{lemm:type_of_zeros} while $\psi_{\bfk}$ has $q$ many $\bar z$-type zeros by \eqref{eq: number_of_zeros} and $p>q$, thus $\varphi_{\overline{\bfk_1}} = \psi_\bfk/Qu_{\bfk - \bfk_1}$ has $p-q>0$ many poles, which contradicts with the fact that $\varphi_{\overline{\bfk_1}}\in \mathcal H_f$. Thus $\dim(\ker_{\mathcal H_f}(D_c(\alpha_1,B)^*+\overline{\mathbf k})) = 0$ for all $\bfk\in \CC$. 
 

In the case that $p=q$, we can assume without loss of generality that $p=q=1$. In this case, we, again, focus on whether $\dim(\ker_{\mathcal H_f}(D_c(\alpha_1, B)^* + \overline{\bfk_1})$ equals to $0$ or not. If there is some $\varphi_{\overline{\bfk_1}}\in \ker_{\mathcal H_f}(D_c(\alpha_1, B) + \bfk_1)$ for some $\bfk_1$, then $\psi_\bfk = \varphi_{\overline{\bfk_1}}(Qu_{\bfk - \bfk_1})\in \ker_{\mathcal H_f}(D_c(\alpha_1, B = 0))$, as already used above or in \eqref{eq: product}. Since $q = 1$, $\psi_\bfk$ has a unique zero. Since $p = 1$, $Qu_{\bfk - \bfk_1}$ has a unique zero (given by \eqref{eq:magnetic_param}). In order to make $\varphi_{\overline{\bfk_1}}\in \mathcal H_f$, the  zeros of $\psi_\bfk$ and $Qu_{\bfk - \bfk_1}$ must coincide. Since zeros of $\psi_\bfk$ and $Qu_{\bfk - \bfk_1}$ are shifted the same way when $\bfk$ varies, by the theta function argument, it is enough to require the two zeros to coincide when $\bfk = 0$, i.e. 
\[
-z_S = -\frac{4\pi }{3\sqrt{3}} = -Z_1 = -\frac{2\pi}{\sqrt{3}} - \frac{4\pi i \bfk_1}{3\sqrt{3}} \qquad \Leftrightarrow \qquad \bfk_1 = \frac{i}{2}. 
\]
Together with the index formula \eqref{eq:index_formula} with $p = 1$, we get for $\bfk_1 \notin \frac{i}{2} + \Gamma^*$, 
\[
\dim \ker(D_c(\alpha_1, B)^*+\overline{\bfk})=0 \text{,~thus~} \dim \ker(D_c(\alpha_1, B)+\bfk) =2,
\]
while when $\bfk \in \frac{i}{2} + \Gamma_{\text{mag}}^*$
\[
\dim\ker(D_c(\alpha_1, B)^*+\overline{\bfk})=1\text{~thus~} \dim \ker(D_c(\alpha_1, B)+\bfk)=3,
\]
which shows that the $2$ flat bands are not gapped. When $\bfk \neq \frac{i}{2} + \Gamma_{\text{mag}}^*$, $\ker(D_c(B, \alpha_1) + \bfk)$ can be written as specified in (2), since the two basis elements are linearly independent as the zeros of the two functions do not coincide. This representation shows readily that the Chern number is $-2$ by using the continuous dependence of the basis on $\alpha_1$. This implies the continuity (in $\alpha_1$) of the projection on the two bands. Using \eqref{theo:away_mag_angl}, we conclude that the Chern number is $-2.$

Let $p<q$. Then as argued above, any element in $\ker(D_c(\alpha_1, B)^*+\overline{\bfk})$ has $q-p$ many $\bar z$ zeros. Thus, using e.g. a direct analog of \cite[Lemma 4.1]{BHZ}, we can construct from each such function precisely $q-p$ flat bands for every $k \in \CC$. The corresponding $\ker(D_c(\alpha_1, B)+\bf k)$ has then dimension $\operatorname{dim}\ker(D_c(\alpha_1, B)+\mathbf k) = \operatorname{ind}_{\mathcal H_f}(D_c(\alpha_1, B)+ \mathbf k) + \operatorname{dim}\ker(D_c(\alpha_1, B)^*+\overline{\mathbf k})$$ = 2p + (q-p)=p+q$. One obtains these $p+q$ flat bands by observing that when multiplying an element of $\ker(D_c(\alpha_1, B))$ ($q$ many $z$-zeros) with an element of $\ker(a)$ ($p$ many $z$-zeros), their product has precisely $p+q$ many $z$-zeros. Hence, we obtain $p+q$ linearly independent flat bands by using theta functions. 

By the Streda formula \cite[(16)]{MK12}, \cite[(3)]{B87}, and \cite[Prop. $3.2$]{BKZ22}, the Chern number can be obtained by differentiating the integrated density of states with respect to the magnetic field, i.e. $$c =- 2\pi \frac{\partial \rho}{\partial B}$$ is the Chern number. Here $\rho(I)$ is the integrated density of states associated with some interval $I$ with $\partial I$ well inside a spectral gap of $H.$
The integrated density of states for a Bloch-Floquet fibred Hamiltonian is equal to the number of bands, normalized by the volume of the fundamental domain. This computation is for example shown in \cite{notes}.

We thus observe that for the integrated density of states of $(D_c(\alpha_1, B)+{\mathbf k})(D_c(\alpha_1, B)^*+\overline{\mathbf k})$, for some $\varepsilon>0$ sufficiently small, we have
 $$\rho([0,\varepsilon))= \frac{q+p}{q \vert \CC/\Gamma \vert} = \frac{1}{\vert \CC/\Gamma \vert} + \frac{p}{q\vert \CC/\Gamma \vert} = 1 + \frac{B}{2\pi}$$ which shows that $c=-1.$

Similarly, for the integrated density of states of $(D_c(\alpha_1, B)^*+\overline{\mathbf k})(D_c(\alpha_1, B)+{\mathbf k})$, we find for some $\varepsilon>0$ sufficiently small
$$\rho([0,\varepsilon))= \frac{q-p}{q \vert \CC/\Gamma \vert} = \frac{1}{\vert \CC/\Gamma \vert}- \frac{p}{q\vert \CC/\Gamma \vert} = 1- \frac{B}{2\pi}$$ which shows that $c=1.$

Thus, the total Chern number is equal to $0.$

\end{proof}
Analogously, we have for two-fold degenerate magic angles.
\begin{theo}[Double magic angles]
\label{theo:double}
     We assume that $\alpha_1 \in \mathcal A$ is a two-fold degenerate magic angle.
     Let the magnetic field be a constant field $B>0$ such that the magnetic flux through $\Gamma$ satisfies $\Phi = \frac{2\pi p}{q}$ for $p \in \mathbb N$ and $q \in \mathbb N$. 
     \begin{enumerate}
     \item If $p>2q$, all $p$ flat bands at zero energy are $A$-lattice polarized
     \begin{equation}
     \begin{split}
    &\ker_{\mathcal H_f}(H_{c,\mathbf k}) =\ker(D_c(\alpha_1,B)+\mathbf k) \times \{0_{\CC^2}\}\\
&\ker_{\mathcal H_f}(D_c(\alpha_1,B)+\mathbf k) =X_{\mathbf k}(\alpha_1).
     \end{split}
     \end{equation} 
     The Chern number of the $p$ flat bands, that are gapped away from the rest of the spectrum, is $-2$. 
     \item If $p=2q$, and thus without loss of generality $p=2q=2$, choose $\bfk_0 \notin \Gamma_{\text{mag}}^*$ and $u_1,u_2$ as in    \eqref{eq: kerD_c}. Then for $\eta_{\bfk-\bfk_0} \in \ker_{\mathcal H_f}(a_{\bfk-\bfk_0})\setminus\{0\}$ and $\eta_{\bfk} \in \ker_{\mathcal H_f}(a_{\bfk})\setminus\{0\}$
     \begin{equation}
     \ker_{\mathcal H_f}(D_c(\alpha_1, B)+\mathbf k) \supseteq \operatorname{span}\{u_1(\alpha_1) \eta_{\bfk-\bfk_0}, u_2(\alpha_1) \eta_{\bfk} \},
     \end{equation}  where equality holds for all points aside from $\bfk =- i.$ 
     The Chern number of the $2$ flat bands is $-2$. The zero energy flat bands are touched by bands from above and below (energetically) at $\bfk =-i.$  All flat bands are still $A$-lattice polarized.   
     \item Let $p<2q$, then there are $4q$ flat bands at zero energy that are gapped from the rest of the spectrum such that for all $\mathbf k \in \CC$
\[    \ker_{\mathcal H_f}(H_{c,\mathbf k})=\ker(D_c(\alpha_1,B)+\mathbf k) \times \{0_{\CC^2}\} 
    \oplus \{0_{\CC^2}\} \times\ker(D_c(\alpha_1,B)^*+\overline{\mathbf k}) \]
    where $\operatorname{dim}\ker(D_c(\alpha_1,B)+\mathbf k) = p+2q$ and $\operatorname{dim}\ker((D_c(\alpha_1,B)+\mathbf k)^*) = 2q-p.$
     The total Chern number of the flat bands is $0.$
     \end{enumerate}
\end{theo}
\begin{proof}
We shall use the same notation as in the proof of Theorem \ref{theo:simple} and just sketch the main differences.

For $p>2q$: The proof of the first result follows as before with now $p-2q$ poles, instead. This is because $\psi_{\bfk}$ has $2q$ many zeros.

For $p=2q=2$: The unique non-zero solution in the nullspace of $D_c(\alpha_1,B=0)^*+\bar{\mathbf k}$, has by \eqref{eq:symmetry} and \eqref{eq:translation_prop}  two $\bar z$-zero summing up to $z_S+\tfrac{4\pi i}{3\sqrt{3}}\bfk$. In this case $Z_2 -\frac{4\pi}{\sqrt{3}}=-z_S$ has the solution $\bfk=-i.$

For $q<2p$: Elements in $\ker(D_c(\alpha_1, B)^*+\overline{\bfk})$ have $2q-p$ many $\bar z$ zeros which gives rise to $2q-p$ flat bands and elements in $\ker(D_c(\alpha_1, B)+{\bfk})$ exhibit $p+2q$ many $ z$ zeros which gives rise to $p+2q$ flat bands.
The Chern numbers coincide with the case of simple magic angles. 

\end{proof}

\section{Spectral properties}
\label{sec:Proofsection}
In this section, we provide a basic spectral analysis of the chiral and anti-chiral magnetic Bistritzer-MacDonald model introduced in \eqref{eq:DcB} and \eqref{eq:achiral}. 

We start by studying the existence of flat bands for magnetic potentials that are periodic with respect to the magnetic moir\'e lattice $\Gamma_{\operatorname{mag}}$. Consider the Hamiltonian $H$ introduced in \eqref{eq:operator}.
We shall use the Floquet operators $H_{\mathbf k}$ as introduced in Section \ref{sec:floquet} for quasi-momenta $\mathbf k \in \mathbb C$ acting on the fundamental cell $\mathbb C /\Gamma_{\operatorname{mag}}$ with periodic boundary conditions. We then introduce the parameter set, of flat bands at energy zero, for the chiral Hamiltonian
\begin{equation}
\begin{split}
 \mathcal A _{\text{c}}&:= \left\{ \alpha_1 \in \CC ; 0 \in \bigcap_{ \mathbf k \in \mathbb C} \Spec_{L^2(\mathbb C / \Gamma_{\operatorname{mag}})}\left(H_{c,\mathbf k}(\alpha_1)\right)\right\}
 \end{split}
 \end{equation}
and denote the analogous set of $\alpha_0 \in \mathbb C$, for the anti-chiral model, by $ \mathcal A _{\text{ac}}$.
Our first theorem shows that in the chiral Hamiltonian, periodic magnetic potentials do not affect the presence of flat bands as characterized in \cite{TKV19,BEWZ20a} and shown to exist in \cite{BEWZ20a,BHZ,LW21}. In contrast to the chiral case, the anti-chiral Hamiltonian \eqref{eq:achiral} does not have any flat bands.

\begin{figure}
\includegraphics[width=7.5cm]{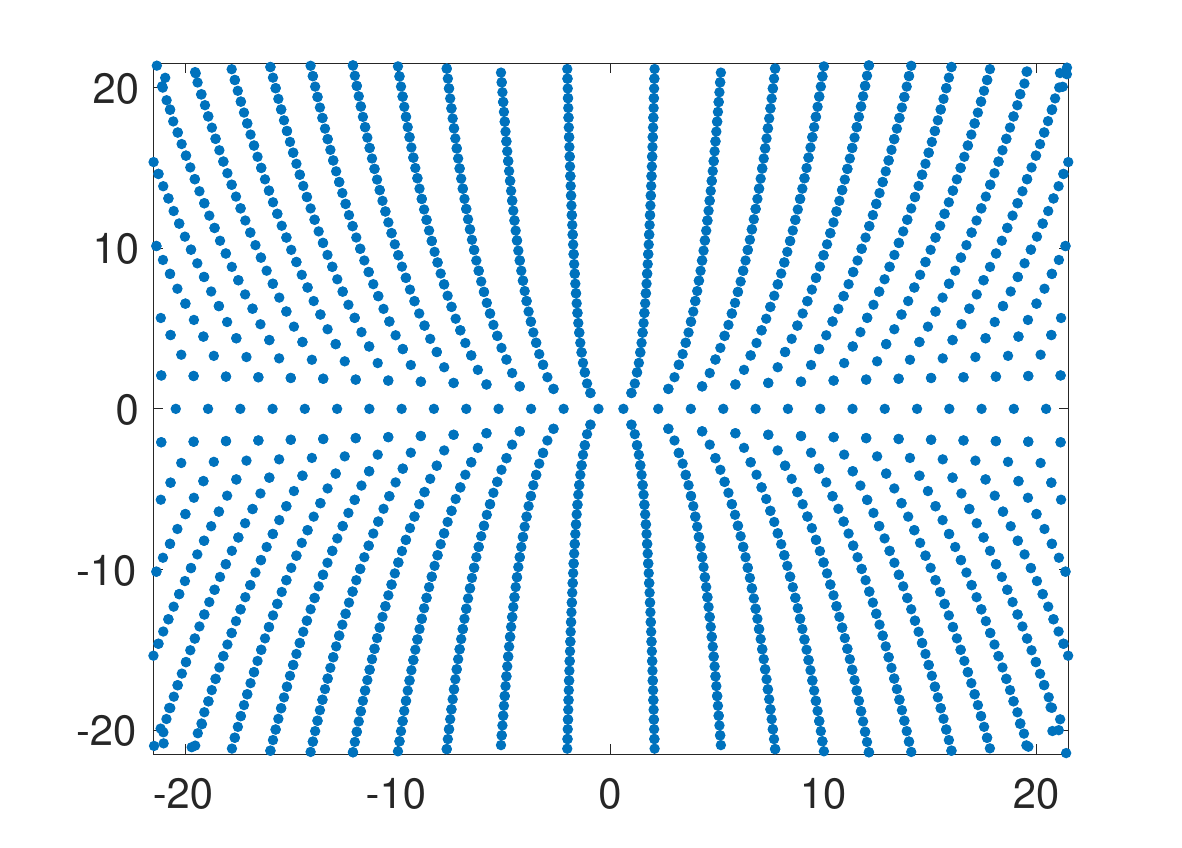}
\includegraphics[width=7.5cm]{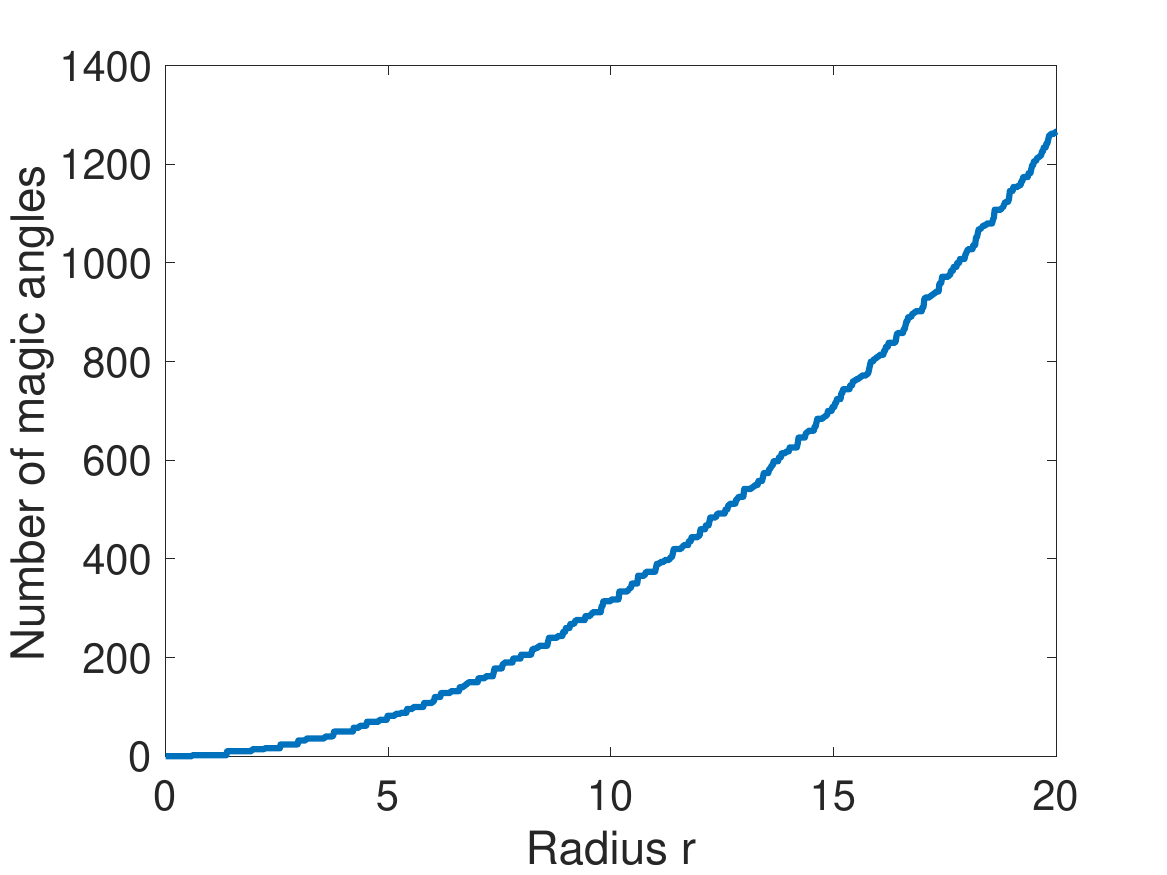}
\caption{Spectrum of $T_{\mathbf k}$ for $\mathbf k \notin \Gamma^*$ with periodic magnetic field. The magic $\alpha_1$ do not depend on the magnetic field strength (left). Number of magic angles within radius $r$ (as in figure on the left) showing quadratic dependence (right), cf. Theorem \ref{theo:magic_angles}.}
\end{figure}
\begin{theo}[Magic angles--Periodic magnetic potentials]
\label{theo:magic_angles}
Consider the BM model with $\Gamma_{\operatorname{mag}}$-periodic magnetic potential $A \in C^{\infty}(\CC/\Gamma_{\operatorname{mag}}; \RR^2)$: \newline 
\underline{Chiral Hamiltonian:} The magic angles, and the multiplicity of flat bands of the chiral Hamiltonian \eqref{eq:DcB}, are independent of the magnetic potential, i.e. $\alpha \in \mathcal A_c$ for $A=0$ if and only if $\alpha \in \mathcal A_c$ for non-zero periodic $A\in C^{\infty}(\CC/\Gamma_{\operatorname{mag}}; \RR^2).$ In addition the number of magic angles, counting multiplicities, in a disc of radius $R$ satisfies: $\vert \mathcal A_c \cap B_R(0)\vert =\mathcal O(R^2).$ \newline
\underline{Anti-chiral Hamiltonian:} The anti-chiral Hamiltonian \eqref{eq:achiral}, with magnetic potentials as above, does not possess any flat bands at zero, i.e. $\mathcal A _{\ach}=\emptyset.$ In particular, its spectrum is purely absolutely continuous. 
\end{theo}

We split the proof of Theorem \ref{theo:magic_angles} on the existence/absence of flat bands into two parts, separating the statement about the chiral Hamiltonian from the statement about the anti-chiral Hamiltonian. We start with a discussion of the chiral Hamiltonian.
\begin{proof}[\text{Proof of Theo. }\ref{theo:magic_angles}, \underline{Chiral Hamiltonian:}]

For the chiral Hamiltonian, with $\alpha_0=0$, it suffices to analyze the nullspaces of the off-diagonal operators \eqref{eq:DcB}. Without loss of generality, we can study the nullspace of $\mathcal D_{\mathbf k,c}(\alpha_1)$ where $0 \in \Spec(\mathcal D_{\mathbf k,c})(\alpha_1) \Leftrightarrow \alpha_1^{-1} \in \Spec(T_{\mathbf k})$
with \emph{Birman-Schwinger operator} $T_{\mathbf k}= (2D_{\bar z}-(A_1(z)+iA_2(z))+\mathbf k)^{-1} \begin{pmatrix}0 & U(z) \\ U_-(z) & 0 \end{pmatrix}$ for $\mathbf k \notin \Gamma_{\operatorname{mag}}^*.$
For any zero mode $\chi_{\mathbf k} \in H^1(\CC/\Gamma_{\operatorname{mag}})$ to $D_{\mathbf k,c}(\alpha_1,B=0)$ it follows that $\psi_{\mathbf k}=\chi_{\mathbf k} \cdot e^{-\phi} \in L^2(\CC/\Gamma_{\operatorname{mag}})$, with $\phi$ as in the proof of Prop. \ref{prop:periodic magnetic field} solves for $\mathcal D_{\mathbf k,c} = D_c $ \[\mathcal D_{\mathbf k,c}(\alpha_1)\psi_{\mathbf k}=\psi_0 \underbrace{\mathcal D_{\mathbf k,c}(\alpha_1,B=0)\chi_{\mathbf k}}_{=0}+ \chi_{\mathbf k}\underbrace{(2D_{\bar z} -(A_1(z)+iA_2(z)))\psi_0(z)}_{=0}=0.\]
This shows that $H_c(B=0)$ possesses a flat band if and only if $H_c(B)$ possesses one, with $B$ induced by a $\Gamma_{\operatorname{mag}}$-periodic magnetic potential $A$.  The characterization of magic angles $\alpha_1$ as reciprocals of eigenvalues of $T_{\mathbf k}$ with $A = 0$ follows then from \cite[Theo. $2$]{BEWZ20a}.  We now utilize the compactness of $T_{\mathbf k}$, with $A = 0,$ to give an upper bound on the number of magic angles.

Indeed, let $z=x_1+ix_2=2i\omega(y_1+i\omega y_2)$, we shall first rewrite $D_{\bar z}$ and $\mathcal V$ in new coordinates $(y_1,y_2).$
 Thus, decomposing for $A_{N}:=\Pi_N T_{\mathbf k},$ with $\Pi_N:L^2(\RR^2/2\pi \ZZ^2; \CC^2) \rightarrow \ell^2(\ZZ^2_{2N+1}; \CC^2)$ such that $\Pi_N\left(\sum_{n \in \ZZ^2} a_{\mathbf n} e^{i \langle y, \mathbf n \rangle}\right) = \{a_{(n_1,n_2)} \}_{\vert n_j \vert \le N}, a_{\mathbf n} \in \CC^2$ and introducing $B_N:=T_{\mathbf k}-A_N,$ we can estimate, by specializing to $\mathbf k=1/2$,
\begin{equation*}
\begin{split}
 \Vert &A_N \Vert_1 \le \Vert \Pi_N  (D_{\bar z}-\mathbf k)^{-1} \Vert_1 \Vert \mathcal V \Vert \le  \sqrt{3} \Vert U \Vert_{\infty} \sum_{\vert m \vert_{\infty} \le N} \tfrac{1}{\Big\vert \left(m_1+ \tfrac{1}{2}\right)^2 + \left(m_1+ \tfrac{1}{2}\right)\left(m_2+ \tfrac{1}{2}\right)+\left(m_2+\tfrac{1}{2}\right)^2 \Big\vert^{1/2}}\\
 &= \sqrt{3}\Vert U \Vert_{\infty}  \Bigg(\sum_{\vert m \vert_{\infty} \le 2} \Big\vert \left(m_1+ \tfrac{1}{2}\right)^2 + \left(m_1+ \tfrac{1}{2}\right)\left(m_2+ \tfrac{1}{2}\right)+\left(m_2+ \tfrac{1}{2}\right)^2 \Big\vert^{-1/2}\\
&\quad +\sum_{2<\vert m \vert_{\infty}\le N } \Big\vert \left(m_1+ \tfrac{1}{2}\right)^2 + \left(m_1+ \tfrac{1}{2}\right)\left(m_2+ \tfrac{1}{2}\right)+\left(m_2+ \tfrac{1}{2}\right)^2 \Big\vert^{-1/2} \Bigg)\\
&\le\sqrt{3}\Vert U \Vert_{\infty}\left(17 + \int_{1/2}^{\sqrt{2}N} \int_0^{2\pi} \tfrac{1}{\sqrt{1+\frac{1}{2}\sin(2\varphi)}} \ d\varphi \right) =\sqrt{3}\Vert U \Vert_{\infty} \left(17 + (\sqrt{2}N-\tfrac{1}{2})7\right) \le C_2N.
 \end{split}
 \end{equation*}
 Similarly, we have that 
\[\Vert B_N \Vert \le \sup_{\vert m \vert_{\infty} >N } \Big\vert \left(m_1+ \tfrac{1}{2}\right)^2 +  \left(m_1+ \tfrac{1}{2}\right)\left(m_2+ \tfrac{1}{2}\right)+\left(m_2+ \tfrac{1}{2}\right)^2 \Big\vert^{-1/2} \Vert U \Vert_{\infty} \le \frac{\Vert U \Vert_{\infty}}{(N-\tfrac{1}{2})} \le \frac{C_1}{N}. \]
 
 Thus, for $\vert \alpha_1 \vert \le R$ we take $N$, large enough, such that $\Vert \alpha_1 B_N \Vert \le \frac{RC_1}{N}<1/2.$ Thus, we may pick $N = \lceil RC_1 /2 \rceil.$ Hence, we can write
$ 1-\alpha_1T_{\mathbf k} = (1-\alpha_1B_N)(1-(1-\alpha_1B_N)^{-1} \alpha_1A_N)$
 and the magic $\alpha_1$'s are the zeros of $f(z) = \det(1-(1-\alpha_1B_N)^{-1} \alpha_1A_N).$
Using the standard bound for Fredholm determinants, we have $\vert f(z) \vert \le e^{2 R \Vert A_N \Vert_1} \le e^{2R C_2 N } \le e^{R^2 C_1 C_2}.$
Hence, as $f(0)=1$, Jensen's formula implies that the number $n$ of zeros of $f$ for $\alpha \in B_0(R)$ is bounded by $n(R) \le \log(2)^{-1} \Big( 4R^2 C_1 C_2 \Big).$
\end{proof}

We now continue by showing that the anti-chiral Hamiltonian does not possess flat bands at energy zero.

\begin{proof}[Proof of Theo.\ref{theo:magic_angles}, \underline{Anti-Chiral Hamiltonian:}]
It is well-known that the singular continuous spectrum for periodic operators of this type is empty. 
Thus, it suffices to exclude embedded point spectrum.
The existence of an embedded eigenvalue implies by Bloch-Floquet theory that there is an open set $\Omega$ such that for all $\mathbf k \in \Omega$ we have $0 \in \Spec(\mathcal K_{\mathbf k}(\alpha_0)) $
where
\[ \mathcal K_{\mathbf k}(\alpha_0):= \begin{pmatrix} \lambda & Q_{\mathbf k}(\alpha_0) \\
Q_{\mathbf k}(\alpha_0)^* & \lambda \end{pmatrix}\text{
with }Q_{\mathbf k}(\alpha_0):=
\begin{pmatrix}
  \alpha_0  e^{-i  \theta/2} V(z)  & \mathscr D_{ z}(\overline{\mathbf{k}},A)    \\
\mathscr D_{\bar z}(\mathbf{ k},A) & \alpha_0 e^{-i  \theta/2}  \overline{V(z)}
\end{pmatrix}\]
and we introduced 
\begin{equation}
\begin{split}
\label{eq:Dequs}
\mathscr D_z(\mathbf{ \bar k},A)=2D_{z} +\mathbf{ \bar k}-(A_1(z)-i A_2(z))\text{ and }
\mathscr D_{\bar z}(\mathbf{ k},A)=2D_{\bar z}+\mathbf k-(A_1(z)+i A_2(z)).
\end{split}
\end{equation}
Since the operator $\mathcal K_{\mathbf k}$ depends analytically on the real and imaginary part of $\mathbf k$ individually, the existence of an embedded eigenvalue implies that the entire band would be flat. This can be easily seen by looking at all possible rays of $\mathbf k$ in $\Omega.$

We shall omit the $\alpha_0$ dependence and set $\theta=0$ to simplify notation.
The formal inverse of $\mathcal K_{\mathbf k}$ is given by $  \mathcal K_{\mathbf k}^{-1}= \begin{pmatrix} \lambda(\lambda^2 - Q_{\mathbf k}Q_{\mathbf k}^*)^{-1} & -Q_{\mathbf k}(\lambda^2-Q_{\mathbf k}^*Q_{\mathbf k})^{-1} \\ -(\lambda^2-Q_{\mathbf k}^*Q_{\mathbf k})^{-1}Q_{\mathbf k}^* & \lambda(\lambda^2-Q_{\mathbf k}^*Q_{\mathbf k})^{-1}  \end{pmatrix}.$
The operator $\RR \ni \mathbf k_1 \mapsto \mathcal K_{\mathbf k}$  for $\mathbf k_2$ fixed, and $\lambda \in \RR$ is a self-adjoint holomorphic family with compact resolvent on $L^2(\mathbb C/\Gamma)$. A flat band would imply that $\mathcal K_{\mathbf k}(\alpha_0)$ is not invertible for any $\mathbf k \in \CC.$
To simplify the analysis, we write
\begin{equation}
\begin{split}
Q_{\mathbf k} &= \underbrace{\sum_{j=1}^2 (D_j +\mathbf k_j-A_j(x))\sigma_j}_{=:{H}_{\text{D},\mathbf k}} + \mathcal V \text{ and }
Q_{\mathbf k}^* = \underbrace{\sum_{j=1}^2 (D_j +\mathbf k_j-A_j(x))\sigma_j}_{=:{H}_{\text{D},\mathbf k}} +\overline{\mathcal V},
\end{split}
\end{equation}
where $\mathcal V=\operatorname{diag}(V,\overline{V}).$ Recall also the Pauli operator ${H}_{\text{P},\mathbf k}$ given as
\begin{equation}
\begin{split}
{H}_{\text{P},\mathbf k}=({H}_{\text{D},\mathbf k})^2 = &\left((D_1+\mathbf k_1-A_1(x))^2+(D_2+\mathbf k_2-A_2(x))^2 \right)-  (\partial_1 A_2 - \partial_2 A_1)(x)\sigma_3.
\end{split}
\end{equation}
In this setting, we have that both $\mathbf k_1,\mathbf k_2$ are real.
Thus, we have for $S(\lambda)={H}_{\text{P},\mathbf k}-\lambda^2$
\begin{equation}
\begin{split}
Q_{\mathbf k}^*Q_{\mathbf k}-\lambda^2 &= (1+\underbrace{ ([Q_{\mathbf k}^*, \mathcal V]+\mathcal V Q_{\mathbf k}^* + \mathcal V^* Q_{\mathbf k}+ \mathcal V\mathcal V^*)S(\lambda)^{-1}}_{=:W_1(\lambda)})S(\lambda)\\
Q_{\mathbf k}Q_{\mathbf k}^*-\lambda^2 &= (1+\underbrace{ ([Q_{\mathbf k}, \mathcal V^*] + \mathcal V^*Q_{\mathbf k}+ \mathcal V Q_{\mathbf k}^*+ \mathcal V\mathcal V^*)S(\lambda)^{-1}}_{=:W_2(\lambda)})S(\lambda).
\end{split}
\end{equation}
We now complexify the real part of $\mathbf k$, which is $\mathbf k_1$, and choose $\mathbf k= \mathbf k_1+ i \mathbf k_2 $ with $\mathbf k_1:=(\mu+ i y)  $, where $\mu, y , \mathbf k_2 \in \RR.$ Since the principal symbol of $Q_{\mathbf k}$ is the Dirac operator and the Pauli operator its square, we find by self-adjointness that $\Vert  S(\lambda)^{-1} \Vert,  \Vert Q_{\mathbf k} S(\lambda)^{-1} \Vert, \Vert Q_{\mathbf k}^* S(\lambda)^{-1} \Vert = \mathcal O(\vert y\vert^{-1}).$

Assuming that there exists a flat band to $\mathcal K_{\mathbf k}$, it follows that
$-1 \in \Spec(W_1(\lambda))\cap \Spec(W_2(\lambda))$ in a complex neighbourhood of $\mathbf k_1 \in \RR$ by Rellich's theorem.
Then \cite[Thm $1.9$]{Ka} implies that for all $\mathbf k_1 \in \CC$ we have $-1 \in \Spec(W_1(\lambda)) \cap \Spec(W_2(\lambda)).$
 But this is impossible, by the estimates on $\Vert  S(\lambda)^{-1} \Vert$  for $\vert y \vert$ large enough.
\end{proof}


For general magnetic fields, the concept of bands does not apply. Instead, since a flat band for a Floquet operator, corresponds to an eigenvalue of infinite multiplicity of the original operator, one should study such eigenvalues of infinite multiplicity. Then, we have the following result that we split up into one statement on flat bands and one on eigenvalues of infinite multiplicity
\begin{figure}
\includegraphics[height=4cm,width=6cm]{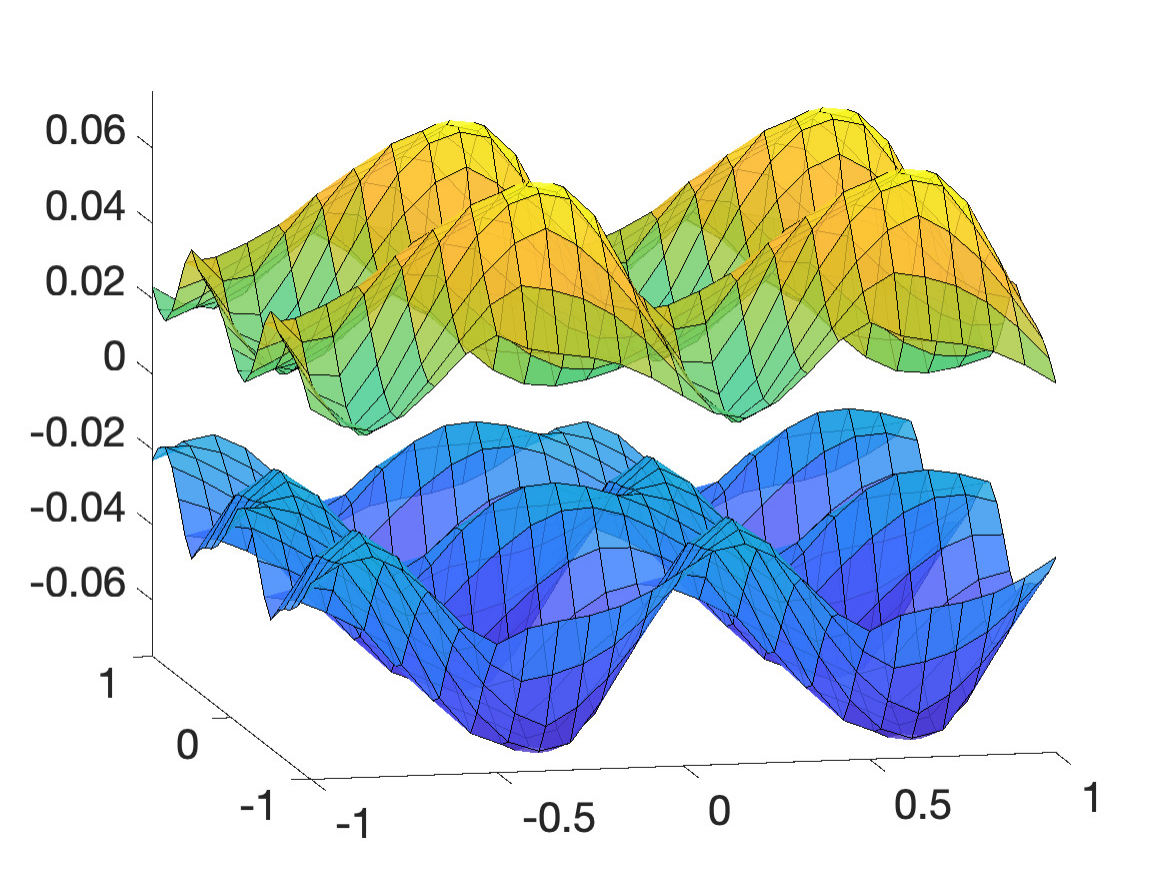}
\includegraphics[height=4cm,width=6cm]{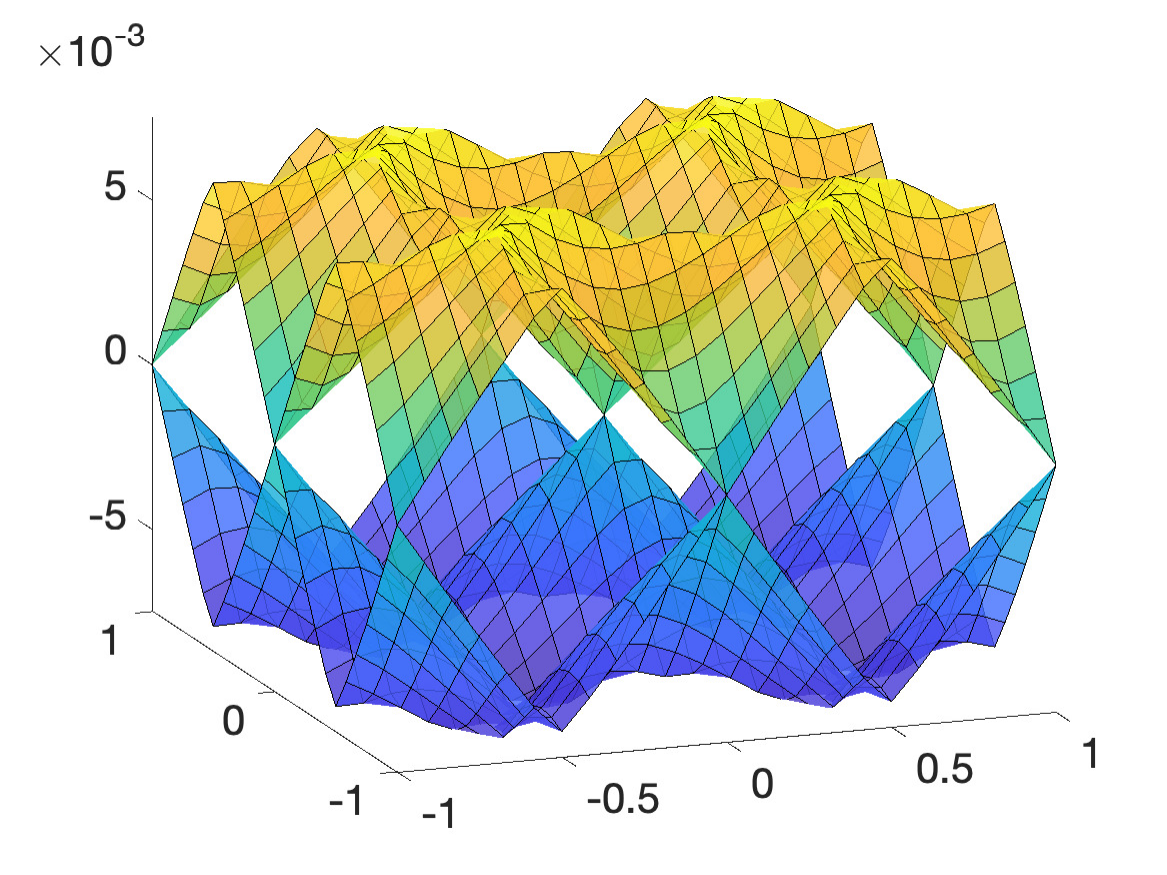}
\caption{Periodic magnetic field $A_1(z) = 2\sqrt{3} \cos(\Im(z))$: On the left, the lowest bands of the anti-chiral Hamiltonian, $\alpha_0=1$, where $0$ is not protected, under periodic magnetic perturbations, on the right the lowest bands of the chiral Hamiltonian, $\alpha_1=1$, where $0$  is protected in the spectrum.}
\end{figure}
\begin{theo}[Eigenvalues]
\label{theo: eigenvalues}
Let $\alpha_1\in \mathcal A$ be a magic angle. When adding to $H_c(B=0)$ any magnetic field $B \in L^{\infty}_{\operatorname{comp}}$ with flux $\lfloor\Phi/2\pi\rfloor \ge 1$\footnote{We let $\lfloor y\rfloor$ be the largest integer \emph{strictly} less than $y$.} or any periodic magnetic potential $A_{\operatorname{per}} \in C^{\infty}(E),$ the operator $H_c(B)$ has an eigenvalue of infinite multiplicity at $0$. If $\alpha_1$ is not magic, and $B \in L^{\infty}_{\operatorname{comp}}$ as above, then $H_c$ possesses an eigenvalue of multiplicity $\lfloor\Phi/2\pi\rfloor$ at zero. In particular, for non-zero constant magnetic fields the chiral Hamiltonian possesses an eigenvalue of infinite multiplicity at zero for any $\alpha_1 \in \mathbb R.$
\end{theo}
\begin{proof}
To see that $0$ is in the spectrum of the chiral magnetic Hamiltonian, we use that we can multiply any $\psi=(\psi_1,0) \in \CC^4$, with $D_c(B=0)\psi_1=0$ by a $\varphi_1$, satisfying $(2D_{\bar z}-A(z))\varphi_1=0,$ to define the new function $\chi=(\varphi_1 \psi_1,0)$ which then satisfies $H_c(B) \chi=0.$ While there also exists a solution $D_c(B=0)^*\psi_2=0$, there does not exist a zero mode to the operator $2D_z-\overline{A(z)}$ for $B>0.$

By the Aharonov-Casher effect \cite[Sec.$6.4$]{CFKS}, there are precisely $\lfloor \Phi/(2\pi) \rfloor$ linearly independent square-integrable zero modes $\varphi_1,...,\varphi_{\lfloor \Phi/(2\pi) \rfloor}$. Multiplying this with the Floquet-periodic zero modes of the chiral model $\psi$, which exist for all $\alpha_1 \in \mathbb C$, this gives the claim for the magnetic fields of compact support. When $\alpha_1$ is magic, the non-magnetic Hamiltonian exhibits an eigenvalue of infinite multiplicity at energy zero. Thus, there exists a countable family of zero modes $(\psi_n)_n$ to $D_c(B=0)$ at $\alpha_1$ magic. Thus, we can construct a countable family of zero modes $\chi_n:=\varphi_1 \psi_n$ to $D_c(B)$ at $\alpha_1$ magic.   

Turning to constant magnetic fields with flux $\Phi \in 2\pi \mathbb Q$, through a fundamental domain $E$ for some $n$, then by adding potentials $A_{\operatorname{per}} \in C^{\infty}(E)$, there is by Proposition \ref{prop:periodic magnetic field} a $\varphi_{\mathbf k}$ such that $(a_k+A_{\operatorname{per}})\varphi_{\mathbf k}=0$ and the existence of a flat band at zero follows. If the fields are not commensurable, the same argument shows the existence of an eigenvalue of infinite multiplicity.

The stability under perturbations by periodic magnetic potentials, follows directly from the existence of periodic $\psi_0\neq 0$ such that, $(2D_{\bar z}+A_{\operatorname{per}}) \psi_0=0.$
\end{proof}

\section{H\"ormander condition and exponential localization of bands}
\label{sec:loc_bands}
In this section, we study the exponential squeezing of bands for periodic magnetic fields and small angles. In particular, we shall see that in the chiral model, there will be at least $\sim 1/\theta$ many bands in an exponentially (in $\theta$) small neighbourhood around zero. We deduce this property by studying the existence of localized quasi-modes in phase space. Phrased differently, for small twisting angles any angle \emph{wants to be magic}. We shall prove this for the chiral model and then show that in the anti-chiral model such quasi-modes do not exist. In the case of the non-magnetic BM Hamiltonian, this has been established in \cite{BEWZ20a,BEWZ20b}.

\subsection{Exponential squeezing in chiral model}
\label{SS:Expsqueeze}

The chiral model possesses in general a lot quasi-modes located close to the zero energy level. Indeed, since $h=1/B$ is our semiclassical parameter, the principal symbol of $hD_c^{\theta}$, with $D_c^{\theta}$ as in \eqref{eq:DcB}, is just  $p(z,\zeta):=\sigma_0(hD_c^{\theta})(z,\zeta)=2\bar \zeta-A(z)$. The existence of localized modes will depend on the vanishing/non-vanishing of the bracket 
\begin{equation}
\label{eq:PB}
\{p,\bar p\}(z)=2(\partial_{ z}A(z)-\partial_{\bar z}\overline{A(z)})=2i B(z).
\end{equation}
We observe that with our quantization, the principal symbol and consequently the Poisson bracket are independent of the potentials. To see the effect of the potentials, one may look at the non-equivalent tight-binding limit
\begin{equation}
\label{eq:Dctheta}
 \mathcal D_{c,\text{TB}}^{\theta} = \begin{pmatrix} 2\theta D_{\bar z} - A(z)&  U(z) \\  U_-(z) & 2\theta D_{\bar z} - A(z)\end{pmatrix}.
 \end{equation}
The semiclassical principal symbol of $ \mathcal D_{c,\text{TB}}^{\theta}$ is given by 
\[ \sigma_0( \mathcal D_{c,\text{TB}}^{\theta})(z,\zeta) = \begin{pmatrix} 2  \bar \zeta -A(z) &  U(z) \\  U_-(z) & 2  \bar \zeta -A(z)\end{pmatrix}.\]
The determinant of the principal symbol of $ \mathcal D_{c,\text{TB}}^{\theta}$ and its conjugate symbol is given for $W(z):= U(z)U_-(z)$ by $q(z,\zeta):=\left(2\bar \zeta - A(z) \right)^2 - W(z).$

We then have the following existence of quasimodes result, which for the semiclassical scaling in $\mathcal D_{c}^{\theta}$ follows from the Poisson-bracket \eqref{eq:PB} along the lines as presented for $\mathcal D_{c,\text{TB}}^{\theta}$ below. 
 \begin{prop}
 \label{quasimode}
There exists an open set $\Omega\subset \mathbb C$ and a constant $c$ such that for any $\mathbf k\in \mathbb C$ and $z_0\in \Omega$, there exists a family $\theta\mapsto \mathbf u_{\theta}\in C^\infty(\mathbb C/\Gamma;\mathbb C^2)$ such that for $0<\theta<\theta_0$,
 \begin{equation}
 \label{eq:quasimode}|(\mathcal D_{c,\text{TB}}^{\theta}-\theta\mathbf k)\mathbf u_{\theta}(z)|\le e^{-c/\theta},\quad \Vert \mathbf u_\theta\Vert_{L^2}=1,\quad |\mathbf u_\theta(z)|\le e^{-c|z-z_0|^2/\theta}.
 \end{equation}
 \end{prop}
 \begin{proof} 
Since the Poisson-bracket in complex coordinates reads
 \begin{equation}
 \begin{split}
 \{ q_1, q_2 \} = \partial_\zeta q_1 \partial_z q_2 + 
 \partial_{\bar \zeta} q_1  \partial_{\bar z} q_2 - 
\partial_z q_1 \partial_\zeta q_2  -
\partial_{\bar z} q_1 \partial_{\bar \zeta}  q_2,
 \end{split}
 \end{equation}
 we find that under the constraint that $q=\bar q=0$ at some point $(z,\zeta)$ 
 \begin{equation}
 \begin{split}
 \{ q, \bar q \}(z,\zeta) &=  
(\partial_{\bar \zeta} q \partial_{\bar z} \bar q - 
\partial_z q \partial_\zeta \bar q )(z,\zeta)
=8 i \vert W(z) \vert B(z)- 8i \Im( \partial_z W(z) \overline{W(z)}^{1/2}).
 \end{split}
 \end{equation}
 We then have that using that $U(z)=0$
 \[ W(z) = -z^2  (\partial_z U(0))^2(1+ \mathcal O(\vert z \vert)) \text{ and } \partial_z W(z) =-2z  (\partial_z U(0))^2(1+ \mathcal O(\vert z \vert)),\]
 which shows that 
 \[ \overline{W(z)}^{1/2} \partial_z W(z) = 2i  \vert z \vert^2 \vert \partial_z U(0) \vert^2 \partial_z U(0) (1+\mathcal O(\vert z\vert)).\]
 This implies the following expansion of the Poisson bracket at zero
 \[ \{q,\bar q\} = 8i  \vert z\vert^2 \vert \partial_z U(0)\vert^2 (B(z)-2 \Re( \partial_zU(0)))(1+\mathcal O(\vert z \vert)).\]

The result then follows from a real-analytic version \cite[Theorem 1.2]{DSZ04} of H\"ormander's local solvability condition:
For a differential operator $ Q = \sum_{ |\alpha| \leq m } 
a_\alpha ( x, \theta ) ( \theta D)^\alpha $ with real-analytic maps
$ x \mapsto a_\alpha ( x, \theta ) $ near some $ x_0 $, we let
$ q ( x, \xi ) $ be the semiclassical principal symbol of $ Q $. If for phase space coordinates $(x_0,\xi_0)$ we have $ q ( x_0 , \xi_0 ) = 0 , \ \  \{ q , \bar q \} ( x_0 , \xi_0 ) \neq 0 , $
then there exists a family $ v_\theta \in C^\infty_{\rm{c}} ( \Omega ) $, $ \Omega $ is a
neighbourhood of $ x_0 $, such that for some $ c > 0$
\begin{equation}
\label{eq:quas}
| (\theta D)^\alpha_x Q  v_\theta (x ) | \leq C_\alpha e^{ - c / \theta } , \ \
\| v_\theta \|_{L^2}  = 1, \ \ | (\theta \partial_x)^\alpha v_\theta ( x ) | \leq C_\alpha e^{ - c | x- x_0|^2/ \theta }.
\end{equation}
\end{proof}

We then have the following result exhibiting the exponential squeezing of bands:
\begin{theo}[Exponential squeezing of bands]
\label{theo:exp_squeez_bands}
Consider the semiclassical scaling of the chiral Hamiltonian with magnetic potential $A \in C^{\infty}(E)$ inducing a non-zero magnetic field or consider the chiral Hamiltonian with tight-binding scaling \eqref{eq:Dctheta} and arbitrary magnetic potential $A \in C^{\infty}(E)$. For the Floquet-transformed operator, the spectrum is a union of bands
\[ \Spec_{L^2(E)} (H_{\mathbf k}^{\theta}) = \{ E_j(\mathbf k,\theta)\}_{j \in \ZZ}, E_j(\mathbf k,\theta ) \le E_{j+1}(\mathbf k,\theta), \mathbf k \in \CC,\]
where $E_0(\mathbf k,\theta) = \min_j \vert E_j(\mathbf k,\theta) \vert.$ 
Then there exist constants $c_0,c_1,c_2>0$ and $\theta_0>0$ such that for all $\mathbf k \in \CC$ and $\theta\in (0,\theta_0)$, 
\[ \vert E_j(\mathbf k,\theta) \vert \le c_0e^{-c_1/\theta}, \vert j \vert \le c_2 \theta^{-1}. \]
\end{theo}
\begin{proof}
By using the above proposition and \cite[Prop. $4.2$]{BEWZ20a}, we can use the proof of \cite[Theo. $5$]{BEWZ20a} to deduce the result.
\end{proof}

\subsection{Anti-chiral model}
To see that the conclusion of Theorem \ref{theo:exp_squeez_bands} does not hold for the anti-chiral Hamiltonian, we proceed as follows, and shall again restrict us to the slightly more technical tight-binding scaling.
Consider in \eqref{eq:achiral} for small $\theta>0$ the operator, with $b_{\mathbf k}:=2\theta D_{ z} - \overline{A(z)}-\theta\overline{\mathbf k}$,
\[ \mathcal D_{\ach,\text{TB}}^{\theta}= \begin{pmatrix}  V & b_{\mathbf k} \\ b_{\mathbf k}^*&  \overline{V} \end{pmatrix}. \]
Then, the existence of a zero mode is equivalent, for $z \in \Omega:=\{ w \in \CC; V(w,\bar w)\neq 0\} $, to a zero mode of the operator
\[P(\theta)v_{\theta}(z) = \left(V(z)b^*_{\mathbf k}(V(z)^{-1} b_{\mathbf k} - \vert V(z) \vert^2 \right)v_{\theta}(z).\]
We then find
\begin{prop} If $u(\theta)$ is smooth on a bounded domain with $\WF_h(u(\theta))=\{(z_0,\zeta_0)\} \in T^* \Omega,$ then it follows that $\Vert u(\theta) \Vert \le \frac{C}{\theta} \Vert P(\theta)u(\theta)\Vert, \theta \downarrow 0.$
Thus, there do not exist any quasi-modes $P(\theta) u(\theta) =\mathcal O(\theta^{\infty}).$
\end{prop}
\begin{proof}
Since $\sigma_0(p)$ is real-valued, the condition $\partial p \neq 0 \text{ on } \{p=0\}$ precisely means that $p$ is of real principal type which implies the result by \cite[Theo. $12.4$]{Zw12}.
\end{proof}
The principal symbol of $P(\theta)$ is given by $\sigma_0(P(\theta))(z,\bar z,\zeta,\bar\zeta)=|2\zeta-(A_1-iA_2)|^2-\alpha_0^2|V(z,\bar z)|^2.$
This is of real principal type since real valued and $\sigma_0(P(\theta))=0$ implies $\partial_\zeta\sigma_0(P(\theta))=4(2\bar\zeta-(A_1+iA_2))\neq 0$ assuming $V(z,\bar z)\neq0$ for all $z\in\mathbb C$. By the proposition above, $\Vert P(\theta)u(\theta)\Vert$ is bounded below by function of order $\theta$. In particular, \eqref{eq:quasimode} does not hold.

\subsection{Landau levels in the absence of a magnetic field}

In the preceding subsection, we showed that for small angles there is an abundance of almost flat bands at small twisting angles close to zero energy. In this subsection, we try to address the question: Do similar quasimodes exist at non-zero energy levels?  

\begin{figure}
    \centering
    \includegraphics[width=5cm]{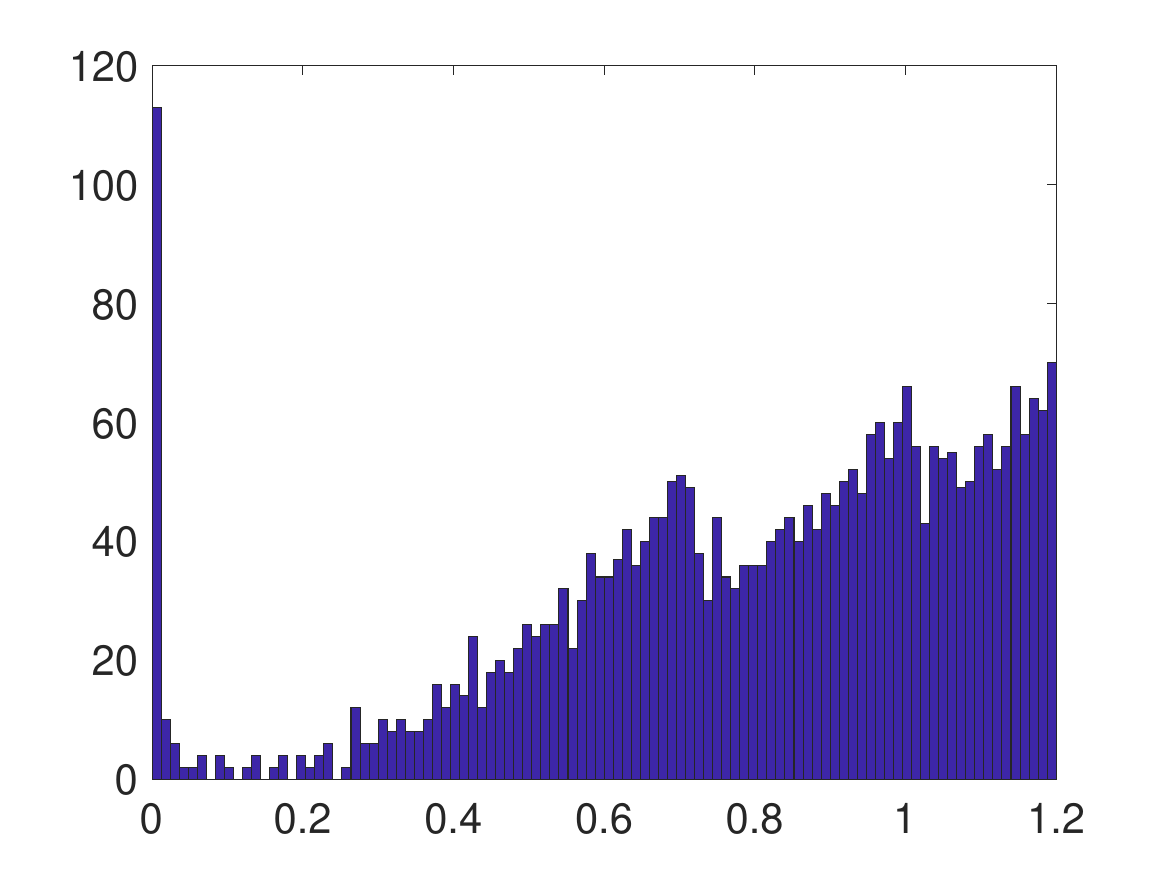}\includegraphics[width=5cm]{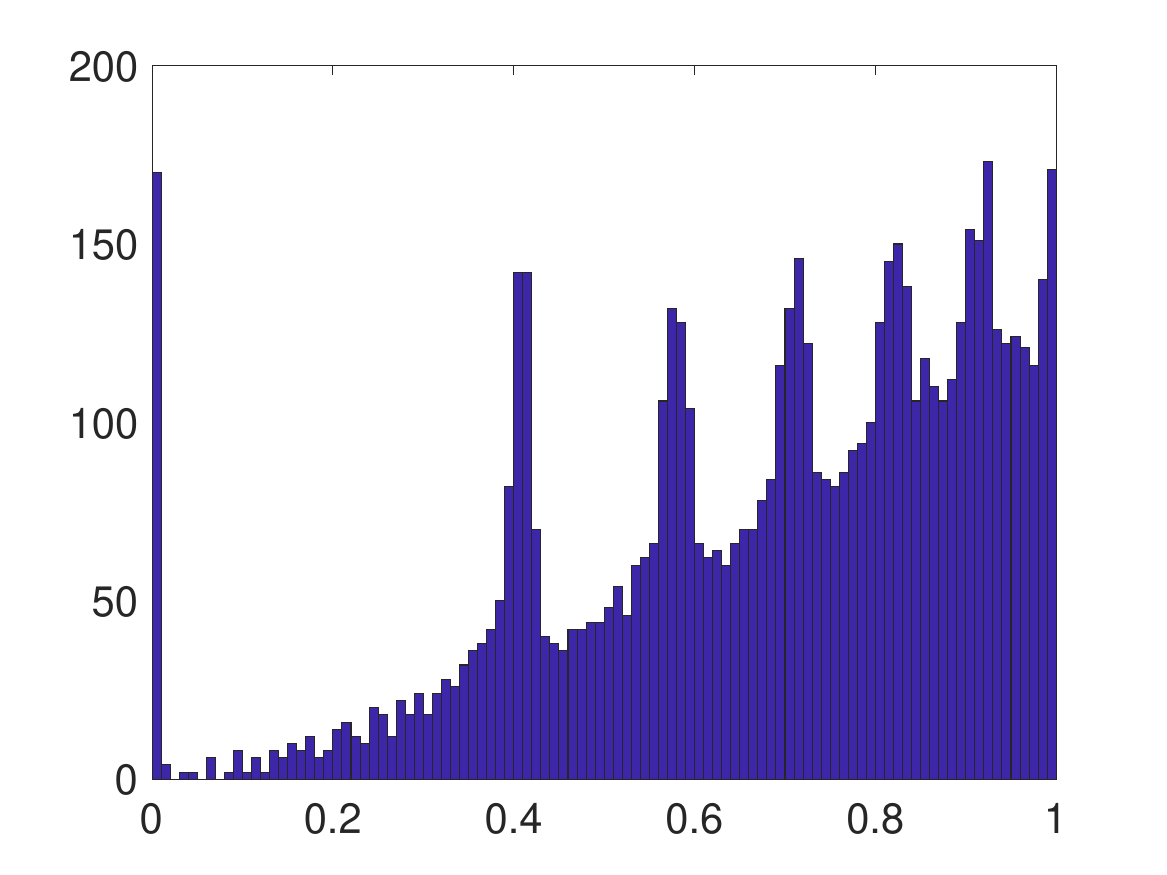}\includegraphics[width=5cm]{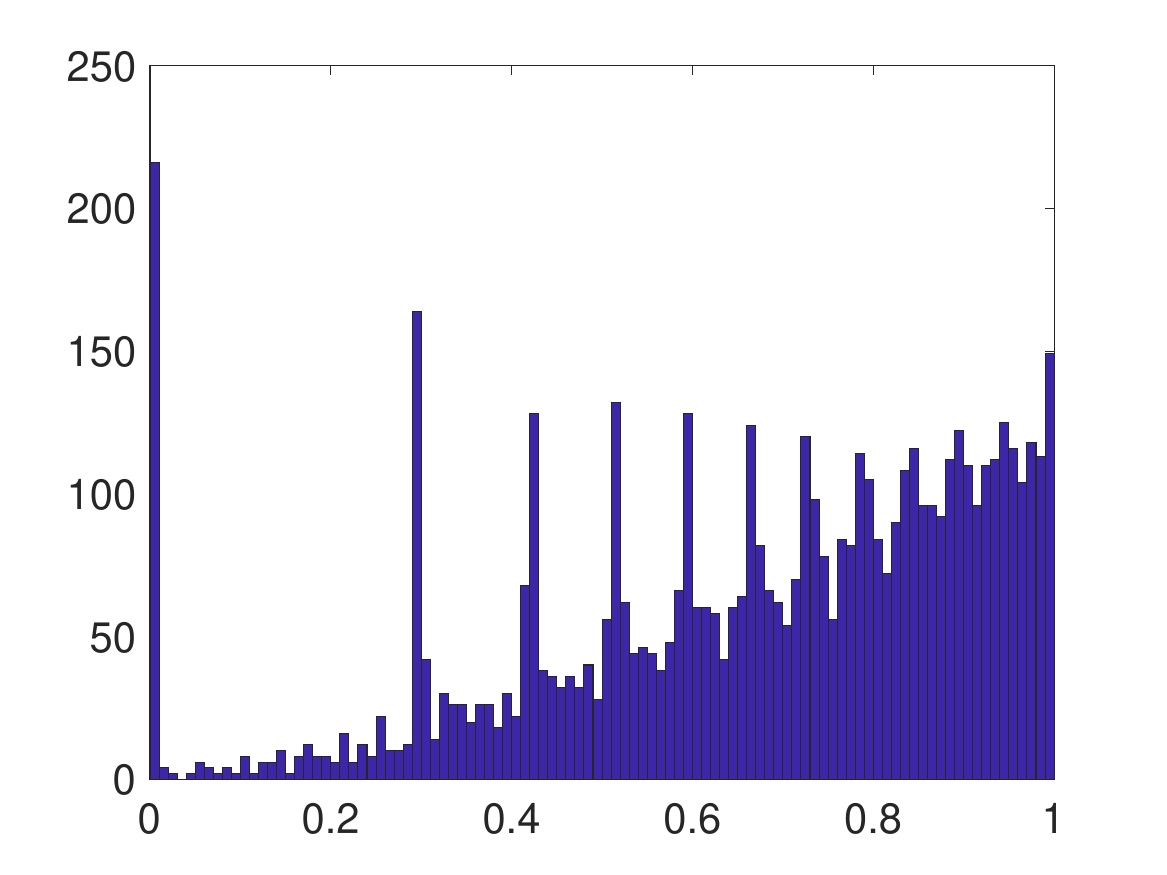}
    \caption{Histogram of Bloch eigenvalues (x-axis Energy/$\sqrt{3}$, y-axis Abundance) of chiral Hamiltonian for fixed $k=1/2$ and $h=0.03$, $h=0.01$, and $h=0.005$, respectively \emph{in the absence of magnetic fields}. The Landau levels are clearly visible with the one at zero energy most significantly pronounced.}
    \label{fig:my_label}
\end{figure}

 It has been proposed, in the absence of external magnetic or electric fields, that there is a close connection between electromagnetic Dirac operators and the chiral/anti-chiral models. Recall that in the chiral limit, we need to consider, in the semiclassical limit $h$, which is closely connected to the study of small twisting angles, the off-diagonal block operator 
 \[
 D = \begin{pmatrix} 2 hD_{\bar z}- A(z) & U(z) \\ U(-z) & 2hD_{\bar z}- A(z) \end{pmatrix}.
 \]
 By conjugating $D$ with $\mathscr U = \frac{1}{\sqrt{2}} \begin{pmatrix} i & 1 \\ -i & 1 \end{pmatrix}$,
 we yield a matrix, 
 $$\mathcal D(z,hD_{\bar z})= \mathscr U D\mathscr U^* = \begin{pmatrix}  D_h^A(z)+ \frac{i}{2}(U(z)-U(-z)) & \frac{i}{2}(U(z)+U(-z)) \\ -\frac{i}{2}(U(z)+U(-z)) & D_h^A(z) - \frac{i}{2}(U(z)-U(-z)) \end{pmatrix},$$
where $D_h^A(z) =2hD_{\bar z} - A(z)$. Here $U(z)-U(-z) =2i \sum\limits_{k=0}^2 \omega^k \sin(\Im(z\bar \omega^k)) = 3 z -\frac{3z \bar z^2}{8}+ \mathcal O(\vert z \vert^5) $ has a Taylor expansion containing only odd products of $z^{n_1}\bar z^{n_2}$, $n_1+n_2 \in 2\mathbb N_0+1$. Similarly, $U(z)+U(-z) =2 \sum_{k=0}^2 \omega^k \cos(\Im(z\bar \omega^k))= \frac{3}{4}\bar z^2 + \mathcal O(\vert z \vert^4) $ has a Taylor expansion containing only even products of $z^{n_1}\bar z^{n_2}$, $n_1+n_2 \in 2\mathbb N.$
By a simple change of variables in $\mathcal D$, $z \mapsto \sqrt{h}z$, we find in terms of $W_1(z)=\tfrac{i}{2}( U(\sqrt{h}z)-U(-\sqrt{h}z) )$ and $W_2(z)=\tfrac{i}{2}(U(\sqrt{h}z)+U(-\sqrt{h}z))$
\[ D(\sqrt{h}z,\sqrt{h}D_{\bar z}) = \begin{pmatrix} D_{\sqrt{h}}^A(\sqrt{h}z) + W_1(z) & W_2(z)\\  -W_2(z) & D_{\sqrt{h}}^A(\sqrt{h}z) - W_1(z)\end{pmatrix}.\]
\[ \begin{split}
    P &=\operatorname{diag}(D(\sqrt{h}z,\sqrt{h}D_{\bar z})^*D(\sqrt{h}z,\sqrt{h}D_{\bar z}),D(\sqrt{h}z,\sqrt{h}D_{\bar z})D(\sqrt{h}z,\sqrt{h}D_{\bar z})^*).\end{split}
\] This operator is a block-diagonal $4 \times 4$ matrix where the first $2\times 2$ block $P_2:=\begin{pmatrix} P_{2,11} & P_{2,12} \\ P_{2,12}^*&  P_{2,22} \end{pmatrix}$ is given by
\[ \begin{split}
   P_{2,11}&=(D_{\sqrt{h}}^A(\sqrt{h}z)+W_1)(D_{\sqrt{h}}^A(\sqrt{h}z)^*+\overline{W_1})+ \vert W_2 \vert^2, \\
   P_{2,12}&=-(D_{\sqrt{h}}^A(\sqrt{h}z)+W_1)\overline{W_2}+W_2(D_{\sqrt{h}}^A(\sqrt{h}z)^*-\overline{W_1}), \text{ and } \\
   P_{2,22}&=(D_{\sqrt{h}}^A(\sqrt{h}z)-W_1)(D_{\sqrt{h}}^A(\sqrt{h}z)^*-\overline{W_1})+ \vert W_2 \vert^2.
\end{split}\]

The leading order expansion in $h$ suggests that the low-lying spectrum of $P$ is determined by the leading-order expansion 
\[ P\varphi = \operatorname{diag}(P_1^{\text{lin}},P_2^{\text{lin}})\varphi + \mathcal O_{\varphi}(h^2),\]
valid for any fixed Schwartz-function $\varphi \in \mathscr S(\mathbb C)$, with
\[P_{1}^{\text{lin}} = h \operatorname{diag}(a_Ba_B^*, b_B b_B^*) \text{ and } P_{2}^{\text{lin}} 
= h\operatorname{diag}(a_B^*a_B,b_B^*b_B)\]
where
\[ a_B = 2D_z - \frac{(3-B_0)i}{2} \bar z \text{ and }b_B = 2D_z - \frac{(3+B_0)i}{2}.\]
In particular, the two components satisfy the commutation relations
\[ [a_B,a_B^*]=(6-2B_0)=2B_a \text{ and }[b_B^*,b_B] = (6+2B_0)=2B_b.\]
Here, we introduced $B_a:=(3-B_0)$ and $B_b:=(3+B_0).$ 
Analogously, for the original Hamiltonian, we already find the leading order in $\sqrt{h}$
\[ H^{\operatorname{lin}}(\sqrt{h}z,\sqrt{h}D_{\bar z}) =\sqrt{h}  \begin{pmatrix}0_{\mathbb C^{2 \times 2}} & \operatorname{diag}(a_B,b_B) \\\operatorname{diag}(a_B,b_B)^* & 0_{\mathbb C^{2 \times 2}} \end{pmatrix}.  \]

\begin{prop}
Let $B_a, B_b \neq 0$, then the spectrum of $H^{\operatorname{lin}}$ consists of infinitely degenerate eigenvalues
\[ \operatorname{Spec}(H^{\text{lin}}) = \{\sgn(n_1)\sqrt{2\vert n_1\vert B_a}, \sgn(n_2)\sqrt{2\vert n_2\vert B_b},\text{ with } n \in \mathbb Z^2 \}.\]
In particular, to each eigenvalue one can find a dense set of exponentially localized Schwartz functions.
\end{prop}
\begin{proof}
We start by studying the squared operator

Let $U_{\operatorname{FW}}:=\frac{1}{\sqrt 2}(1+\sigma_3 \sgn(P^{\operatorname{lin}}))$\footnote{The sign of a self-adjoint operator is uniquely defined by the polar decomposition.} be the Foldy-Wouthuysen transform then 
\[ U_{\operatorname{FW}} H^{\text{lin}} U_{\operatorname{FW}}^* = \sqrt{h} \operatorname{diag}(\sqrt{a_Ba_B^*}, \sqrt{b_B b_B^*} ,-\sqrt{a_B^*a_B},-\sqrt{b_B^*b_B}). \]

We assume in the sequel that $B_a>0$ and/or $B_b>0$ where at least one of the two conditions is satisfied for every $B_0$
\begin{equation}
    \label{eq:spectrum}
 \begin{split}
    \operatorname{Spec}(a_B^*a_B) \subset \{0,2B_a,4B_a,...\} \text{ and }\operatorname{Spec}(a_Ba_B^*) \subset \{2B_a,4B_a,...\},  \\
     \operatorname{Spec}(b_B^*b_B) \subset \{0,2B_b,4B_b,...\} \text{ and }\operatorname{Spec}(b_Bb_B^*) \subset \{2B_b,4B_b,...\}.
\end{split}
\end{equation}
The ground state solutions, in the case $B_i>0$, are infinitely degenerate and then given by 
\[ \psi_0^{m_j}(z) = e^{-\frac{B_i}{4} \vert z \vert^2} z^{m_j-1/2}, \text{ with } m_j \in \mathbb N_0 +1/2.\]
Here, $m_j$ is the eigenvalue of the associated angular momentum operator orthogonal to the plane.
The commutation relations can be used to generate higher Landau levels from the zero modes and show the converse inclusion in  \eqref{eq:spectrum}. The explicit form of the higher Landau levels can be found for exampled in \cite{RB17}.
\end{proof}
\begin{figure}
    \centering
    \includegraphics[width=5.2cm]{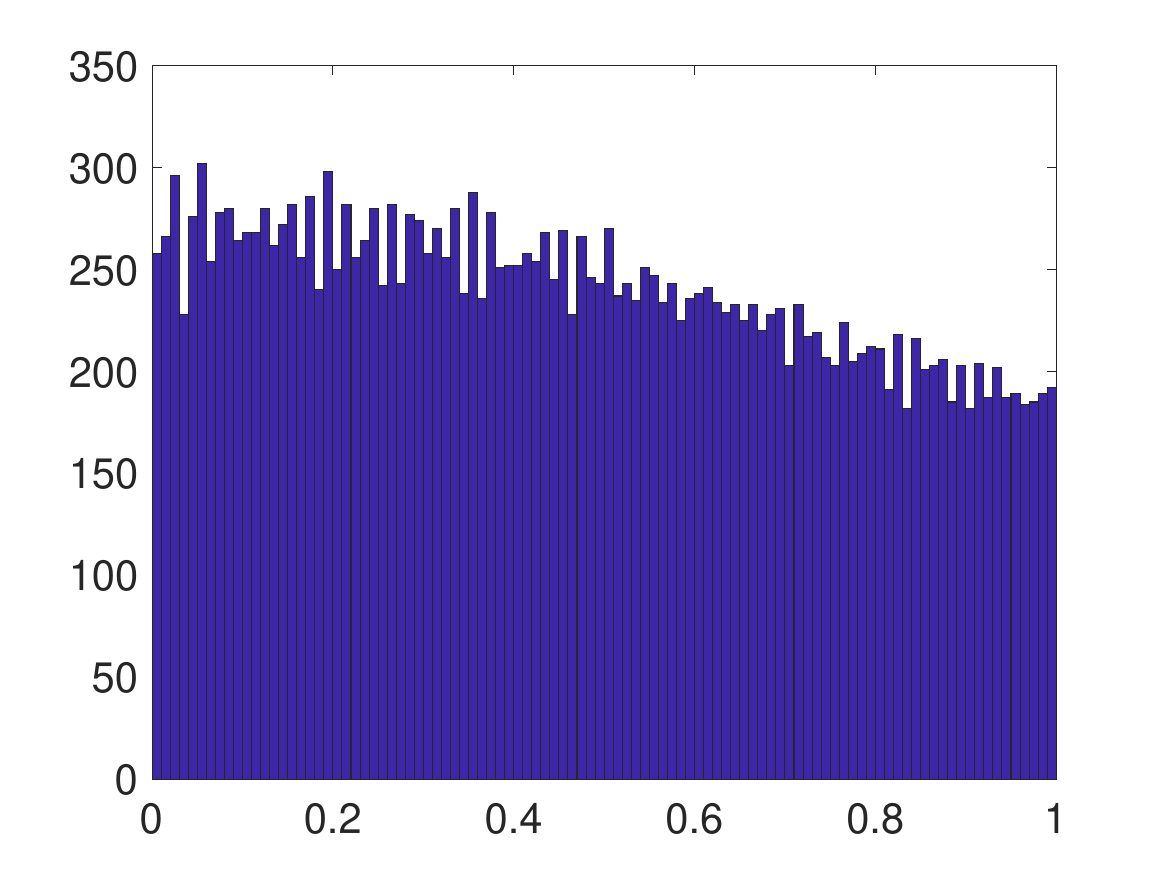}
    \includegraphics[width=5.2cm]{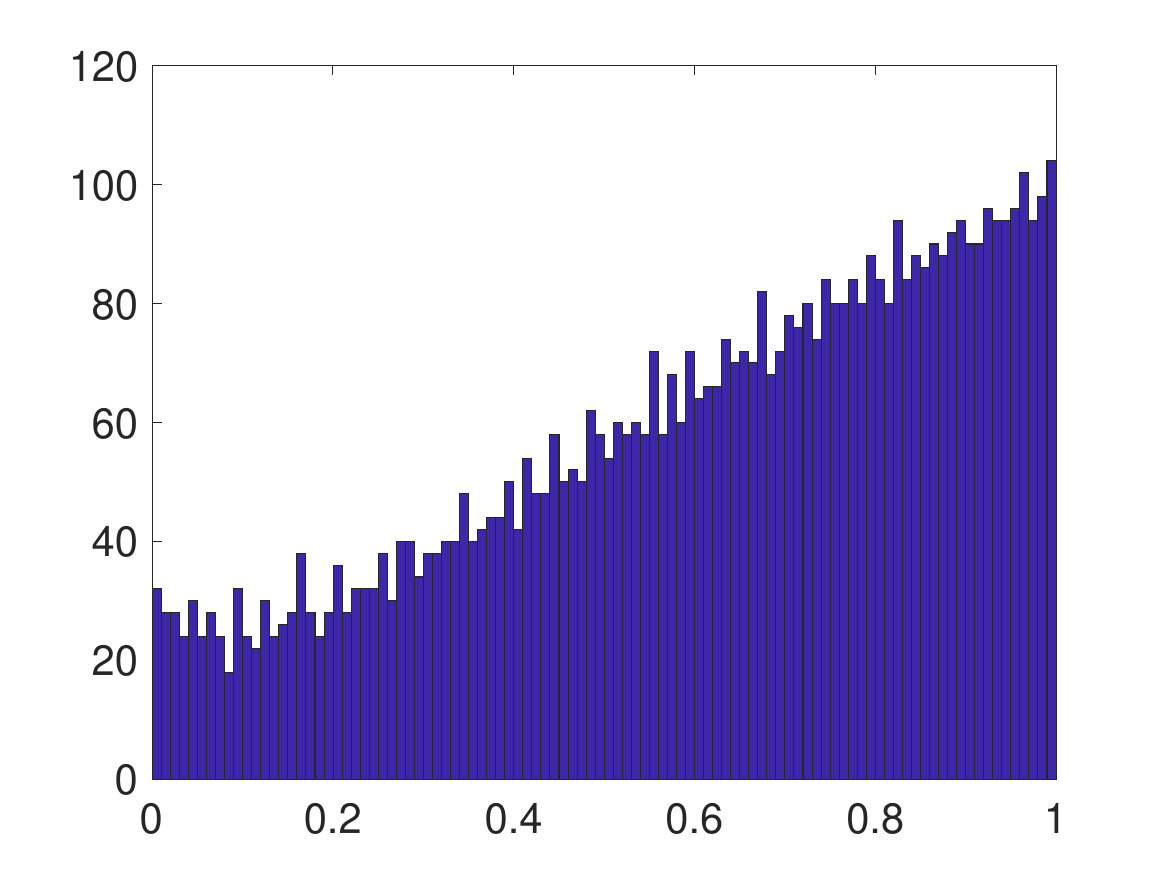}
        \includegraphics[width=5.2cm]{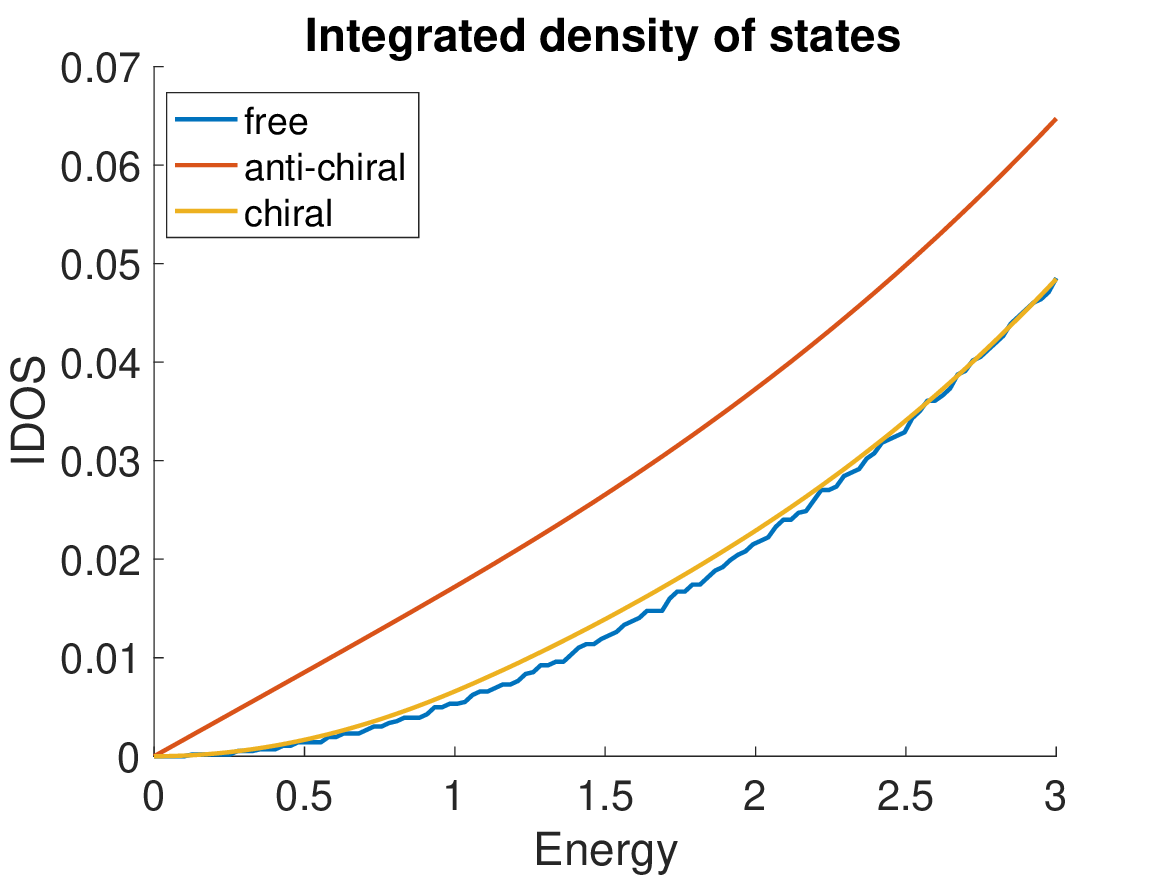}
    \caption{Histogram of anti-chiral model and full Hamiltonian with $h=0.01.$}
    \label{fig:histo}
\end{figure}

\begin{appendix}

\section{Special functions}
\label{sec:Haldane}
Let $\Lambda$ be a $2d$ lattice on $\CC$. The Weierstra{\ss} $\wp$ function is an elliptic function defined by 
\[\wp(z,\Lambda) := \frac{1}{z^2} + \sum_{\gamma\in\Lambda\setminus\{0\}}\left(\frac 1 {(z-\gamma)^2} - \frac 1 {\gamma^2}\right).\]
It is the negative derivative of the Weierstra{\ss} $\zeta$-function
\[\zeta(z,\Lambda) = \frac{1}{z} + \sum_{\gamma \in \Lambda \setminus \{0\}} \Big(\frac{1}{z-\gamma} + \frac{1}{\gamma}+\frac{z}{\gamma^2}\Big), \quad \text{i.e.}\quad \wp(z,\Lambda) = \zeta(z, \Lambda)'.\]
The Weierstra{\ss} $\zeta$-function is further the logarithm derivative of Weierstra{\ss} sigma function
\[
\tilde{\sigma}(z, \Lambda) = z \prod_{\gamma \in \Lambda\setminus\{0\}} \Big(1-\frac{z}{\gamma} \Big)e^{\frac{z}{\gamma}+\frac{z^2}{2\gamma^2}}, \quad \text{i.e.}\quad \zeta(z, \Lambda) = \frac{d}{dz}\log(\tilde{\sigma}(z,\Lambda)) = \frac{\tilde{\sigma}'(z,\Lambda)}{\tilde{\sigma}(z,\Lambda)}.
\]
We shall now specialize to $\Gamma_\lambda$ and discuss when does $\tilde\sigma$ coincide with $\sigma$. If $\lambda_1 = \lambda_2$, the symmetries of $\Gamma$ imply that 
\[ \tilde{\sigma}(\omega z) = \omega \tilde{\sigma}(z) \text{ and } \overline{\tilde{\sigma}(z)} = \tilde{\sigma}(\bar z).\]
Hence, using the previous relation, we find that
\[\zeta(\omega z) = \bar \omega \zeta(z) \text{ and }\overline{\zeta(z)} = \zeta(\bar z).\]
Defining then 
\[
 \eta_1 := \frac{\zeta(z+v_1)-\zeta(z)}{2} \text{ and }\eta_2 := \frac{\zeta(z+v_2)-\zeta(z)}{2}.
 \]
 We find by setting $z = -v_1/2$ and $z=-v_2/2$, respectively, since $\zeta$ is an odd function,
\[ \eta_1 =  \zeta(v_1/2) \text{ and } \eta_2 = \zeta(v_2/2) = \bar \omega \zeta(v_1/2) .\]
Then by Legendre's relation \cite[(14)]{H18}
\[ \eta_1 v_2-v_1 \eta_2 =\zeta(v_1/2)(v_2-\bar \omega v_1)= \pi i, \]
we find that 
$$\eta_1 =-\frac{3}{8\lambda_1}+i \frac{\sqrt{3}}{8\lambda_1} = \frac{\pi \bar{v_1}}{2A} \text{ and }\eta_2 =\frac{3}{8\lambda_1}+i \frac{\sqrt{3}}{8\lambda_1}=\frac{\pi \bar{v_2}}{2A}.$$
Here $A = \frac{\lambda_1^2 8 \pi^2}{3\sqrt{3}}$ is the area of the fundamental domain. 
Comparing with \cite[(17)]{H18}, we see that the modular-invariant and regular Weiersta{\ss} sigma function coincide. 

For lattice vectors $v_1,v_2$ the Weiersta{\ss} sigma function can then be efficiently implemented numerically using 
\begin{equation}\label{eq: sigma_z_another_form}
    \sigma(z) = v_1 e^{\eta_1 z^2/v_1} \frac{\theta_1(\pi z/v_1\vert \omega)}{\pi \theta_1'(0\vert \omega)} =v_1 e^{\eta_1 z^2/v_1} \frac{\theta_1(\pi z/v_1\vert \omega)}{\pi \theta_1'(0\vert \omega)},
\end{equation}
where 
\[\theta_1(z\vert \tau) = \sum_{n\in \ZZ} e^{\pi i n^2 \tau + 2\pi i n z} \text{ for }\Im(\tau)>0.\]

For the honeycomb lattice, this is
\[\sigma(z) \propto e^{\frac{3\sqrt{3}}{16\pi\lambda_1^2}z^2}\theta_1(z/(4i\omega\lambda_1) \vert \omega).\]

\end{appendix}


\end{document}